\documentclass[9pt, a4paper, journal]{IEEEtran}
\usepackage{tikz}
\usetikzlibrary{shapes.geometric, arrows.meta, positioning}

\tikzstyle{process} = [rectangle, minimum width=3cm, minimum height=1cm, text centered, draw=black, fill=blue!10]
\tikzstyle{optional} = [rectangle, minimum width=3cm, minimum height=1cm, text centered, draw=black, fill=green!20, dashed]
\tikzstyle{arrow} = [thick,->,>=stealth]
\usepackage{amsmath}
\usepackage{amssymb}
\usepackage{graphicx,mwe}
\graphicspath{ {./images/} }
\usepackage{authblk}
\usepackage{amsthm}
\usepackage{float}
\usepackage{faktor}
\usepackage{mathtools}
\usepackage{dsfont}
\usepackage{soul}
\usepackage{bigints}
\usepackage{breqn}
\usepackage{bm}
\usepackage{algorithm}
\usepackage{breqn}
\usepackage{algpseudocode}
\usepackage{setspace}
\usepackage{makecell}
\usepackage{multirow}
\usepackage{soul}
\usepackage{comment}
\usepackage{textcomp}
\usepackage[a4paper, portrait, margin=1in]{geometry}
\usepackage{framed}
\usepackage[flushleft]{threeparttable}
\usepackage[most]{tcolorbox}
\usepackage{xcolor}
\colorlet{shadecolor}{orange!15}
\newtheorem{definition}{Definition}
\newtheorem{claim}{Claim}

\newtheorem{theorem}{Theorem}

\newtheorem{corollary}{Corollary}
\newtheorem{assumption}{Assumption}

\newtheorem{proposition}{Proposition}
\newcommand{\R}{\mathbb{R}}

\usepackage{enumerate}
\usepackage{hyperref}
\usepackage{tikz}
\usetikzlibrary{arrows}
\makeatletter
\newsavebox\myboxA
\newsavebox\myboxB
\newlength\mylenA
\usepackage{graphicx} 
\usepackage{caption,subcaption}

\usepackage{xcolor}
\definecolor{light-gray}{gray}{0.95}
\newcommand{\code}[1]{\colorbox{light-gray}{\texttt{#1}}}

\newcommand{\E}[1]{\mathbb{E}\left[#1\right]}
\newcommand{\Prob}[1]{\mathbb{P}\left(#1\right)}

\newcommand*\xoverline[2][0.75]{%
    \sbox{\myboxA}{$\m@th#2$}%
    \setbox\myboxB\null
    \ht\myboxB=\ht\myboxA%
    \dp\myboxB=\dp\myboxA%
    \wd\myboxB=#1\wd\myboxA
    \sbox\myboxB{$\m@th\overline{\copy\myboxB}$}
    \setlength\mylenA{\the\wd\myboxA}
    \addtolength\mylenA{-\the\wd\myboxB}%
    \ifdim\wd\myboxB<\wd\myboxA%
       \rlap{\hskip 0.5\mylenA\usebox\myboxB}{\usebox\myboxA}%
    \else
        \hskip -0.5\mylenA\rlap{\usebox\myboxA}{\hskip 0.5\mylenA\usebox\myboxB}%
    \fi}
\makeatother
\ifCLASSINFOpdf
\else
\fi

\allowdisplaybreaks

\begin{document}
%
\title{Evaluating Policy Effects through Opinion Dynamics and Network Sampling}
%
%
%
\author{Eugene T.~Y. Ang$^{1}$,  Yong Sheng Soh$^{1,2}$
\thanks{$^{1}$Institute of Operations Research and Analytics, National University of Singapore,
        {\tt\small eugene.ang@u.nus.edu}}%
\thanks{$^{2}$Department of Mathematics, National University of Singapore, 
        {\tt\small matsys@nus.edu.sg}}%
}

\maketitle

\begin{abstract}
An essential aspect of effective policymaking is to regularly consider the population's response or feedback towards a newly introduced policy.  These can come in the form of population surveys or feedback channels, and they provide a simple way to understand the ground sentiment towards a new policy.  Conventional surveying methods implicitly assume that opinions are static; in reality, opinions are often dynamic -- the population will discuss and debate these newly introduced policies among themselves, 
and in the process form new opinions.  In this paper, we pose the following set of questions: Can we understand the dynamics of opinions towards a new policy within the population?
Specifically, can we quantify the evolution of opinions over the course of interaction?
How are these changes affected by the topological structure of the underlying network describing the relationship among the population?  
We investigate these questions using a model where the policymaker is able to select a subset of population to which a policy is initially revealed to.  By selecting the subset of respondents judiciously, the policymaker controls the degree of discussion that can take place among the population.  
Under this model, we quantify the changes in opinions between the empirically observed data post-discussion and its distribution pre-discussion, in terms of the number of selected respondents, as well as the number of connections each respondent has within the population network. 
We conduct a series of numerical experiments over synthetic data and real-world networks.  Our work aims to address the challenges associated with network topology and social interactions, and provide policymakers with a quantitative lens to assess policy effectiveness in the face of resource constraints and network complexities.
\end{abstract}

\begin{IEEEkeywords}
Graph Sampling, Network Dynamics, Policy Evaluation
\end{IEEEkeywords}

%
\IEEEpeerreviewmaketitle

\ifCLASSOPTIONcaptionsoff
  \newpage
\fi



%
\section{Introduction}
Rigorous policy design, analysis and evaluation are essential to good governance \cite{oecd2020}. Apart from evaluating whether a policy has produced the desired results or pinpointing the elements that facilitate or impede the policy's effectiveness, it is also imperative to discern how the population or target audience reacts and adapts to the policy. The level of public reception gives policymakers an indication of public confidence towards the policy \cite{grelle2024and}. Furthermore, it provides indications of possible shortcomings, drawbacks and resistance. By identifying areas where the policy may fall short, policymakers can proactively address these challenges and make necessary adjustments to improve its effectiveness and mitigate adverse impacts on the population. 

The gold standard for assessing the efficacy of policies is to deploy randomized controlled trials (RCTs) \cite{Hariton2018}. Notwithstanding the practical or ethical concerns \cite{goldstein2018ethical}, there are methodological challenges. It is often difficult to isolate the respondents in the population as well as the confounding factors that affect one's behavior and actions. For instance, respondents generally interact among themselves and influence each other in the process.  The result is that the collective may arrive at a consensus that is quite different before the interaction \cite{abrams2012formation,larsen1990asch,sherif1937experimental}, and such behavior may affect the overall policy effectiveness \cite{burstein2003impact}. As an example, the government of a country may wish to implement a vaccination campaign to control the spread of a contagious disease such as COVID-19. Interactions within their social circles can influence individuals' decisions to get vaccinated. To make the situation more complex, the presence of vaccine skeptics may (adversely) affect the vaccine take-up rate in the nation \cite{allen2024quantifying}. Hence, policymakers need to discern and evaluate the initial policy beliefs within the population before influential individuals shape public opinion prematurely.

To evaluate a policy effectively, not only do policymakers need an accurate assessment of the initial beliefs and responses towards the policy, but they also need some understanding of how these opinions might evolve and reform as respondents discuss among themselves.  More often than not, these interactions and fine-grained changes in beliefs are not captured in empirical data due to infrequent data collection processes, and policymakers need to resort to statistical inference to estimate such deviations. Furthermore, a binary response is sometimes not sufficient to capture one's response due to his or her ambivalence \cite{persson2021opinions}. Hence, it is important to design an appropriate policy evaluation study to accommodate a range of responses. Certain evaluation exercises, such as polls and referendums with categorical options, might not be able to capture the nuances in public opinion. One example is the 2016 Brexit referendum in the UK, in which voters had to choose to ``Leave" or ``Remain" within the European Union. Although the referendum produced an outcome on the basis of a simple majority, these blunt options can be difficult to interpret, and voters cannot express their opinions constructively \cite{van2021referendums,wagenaar2019}. 

Besides designing sophisticated studies, policymakers often face resource and feasibility constraints.  Certain information may not be available to these policymakers, such as individuals' covariates like income and age, and network connections, due to legal constraints. Examples of such provisions or laws include the Personal Data Protection Act (PDPA) in Singapore and the General Data Protection Regulation (GDPR) in the European Union and the European Economic Area. Thus, the policymakers might only have access to limited population data to ascertain the public opinion of the policy. Moreover, conducting periodic population-wide consensus surveys is time-consuming and resource-intensive. Realistically, policymakers are only able to survey a small group of people at a single time instance to infer the population's opinion. 

\subsection{Our Contributions}
The main objective of this work is to understand the impact of social dynamics and interaction on the population's opinion towards policies.  Concretely, we build a model that describes a population's response towards a new policy.  As time passes, the population is given the opportunity to discuss the newly announced policy, and may reform their personal beliefs and opinions.  Our central question is to quantify the size of these changes based on the characteristics of the social network underlying the population. 

The policymaker is ultimately aware that the population will discuss these policies among themselves, and may be interested in understanding the population's general response at different stages: one closer to the initial stages where opinions are raw, relatively uninformed, and possibly quite diverse, or in the later post-discussion stage where the policy has been extensively debated and a broader consensus may have been reached. To obtain the opinion polls at different stages, we assume that policymakers are able to select the subset of respondents they wish to reveal the policy to.  By doing so, policymakers have control over the amount of interaction between the respondents.  There are several ways policymakers can do so.  For example, the policymakers can select groups of respondents who do not know one another and are unlikely to interact among themselves, allowing policymakers to estimate initial policy responses free from the effects of interactions \cite{ang2025estimating,karwa2018,fatemi2023network}. However, it is impractical to attempt to recruit groups of respondents who do not know each other, as this assumes some knowledge about existing relationships among the population (e.g., through social media platforms). Policymakers can consider alternative strategies, such as randomly selecting a large group of respondents without regard for the underlying relationships, in the hope that the gains from estimating from a larger sample size are not severely impacted by the effects of discussion among the respondents.  In estimating the initial policy responses, policymakers face a tradeoff between the number of respondents they survey and the impact of discussion among respondents. 

\looseness=-1
In our work, we try to understand how the choice of which subset of respondents a policy is revealed to -- whether it is a small but carefully selected subset of the population who are unlikely to communicate among themselves, or a larger and more representative subset but will inadvertently reform opinions after discussion -- impacts population opinions. A different aspect that also affects opinion dynamics is the nature of the relationship among respondents, ranging from mutual interactions to unilateral influence exerted by highly influential individuals. To this end, we provide quantitative bounds on how these opinions change according to various sampling strategies, network topologies, and interaction contexts.
We describe the interplay between these factors that affect the estimation of the initial policy responses.
We also conduct numerical experiments using synthetic and real-world networks to illustrate the tradeoff between the sample size and belief deviation arising from interaction, on various sampling strategies, network topologies, and interaction contexts.  We illustrate specific instances where either factor dominates the other and describe how the number of relationships relative to the population size plays a critical role in the estimation accuracy. 

\subsection{Paper Layout}
\looseness=-1
In Section \ref{sec:related work}, we discuss past work on modeling opinion dynamics and policy evaluation methodologies used to estimate policy effects. In Section \ref{sec:problemformulation}, we describe the problem formulation and in Section \ref{sec:mainresults}, we provide the main results.  In Section \ref{sec:numericalexperiments}, we conduct numerical experiments across different synthetic and real-world networks, interaction contexts, and initial belief distributions, and we discuss insights drawn from these experiments.
\section{Related Work}\label{sec:related work}
In this section, we discuss existing work on opinion dynamics models as well as methodologies that use observational data to estimate policy effects in networked settings.

\subsection{Effect of Opinion Dynamics}\label{subsec:opinion}
The implementation of a new policy prompts discussions among people, and these interactions influence their beliefs and behaviors regarding the policy. As public opinion plays an important role in determining the effectiveness of the policy \cite{jones2009trans}, it is important to understand and explore how consensus and polarization develop over time.  One can adopt a data-driven approach by analyzing large-scale social data on communication platforms \cite{gonzalez2013social}. However, these data may not be entirely available due to legal or confidentiality issues. Hence, limited access to data as well as complexities arising from the data collection process call for realistic modeling of opinion formation and interaction \cite{peralta2022opinion,sobkowicz2009modelling,castellano2009statistical}. There are several notable opinion dynamics models, such as the voter model \cite{holley1975ergodic,redner2019reality}, DeGroot model \cite{degroot1974reaching} and bounded confidence model \cite{lorenz2007continuous,hegselmann2002opinion,deffuant2000mixing}. These models assign an opinion state variable to every individual. These variables change over time based on the system's stipulated mechanism. The resulting systemic state usually depends on the underlying graphical structure \cite{castellano2009statistical,golub2010naive,xia2011opinion}.  These models help to explain the formation of collective opinions, the shift in beliefs due to peer and social influence and the emergence of social consensus or division. To better reflect real-world opinion dynamics, researchers have developed various methods for estimating the strength of social influence within networks. Recent empirical studies estimate influence weights through controlled experiments that track how individuals adjust beliefs in response to peer confidence \cite{moussaid2013social} or by calibrating opinion dynamic models to observational data \cite{gestefeld2023calibrating}. Machine learning approaches, such as generative adversarial networks (GANs) \cite{wang2025modeling}, graph neural networks (GNNs) \cite{li2025unigo}, and inverse reinforcement learning \cite{yuan2025behavioral}, have also been used to estimate influence patterns from complex social data, offering an alternative to traditional model-based inference.  In our work, we incorporate these mechanisms to understand how surveyed respondents update their beliefs through interactions. By modeling the interactions and opinion dynamics, policymakers can gauge the initial distribution of population beliefs and quantify the evolution of the policy beliefs.

\subsection{Challenges in Policy Evaluation Methods}
It is challenging to assess a policy's effectiveness, as individuals' behaviors and beliefs are largely influenced by their social connections. This influence potentially confounds the evaluation of policy effects and complicates the assessment of its overall effectiveness. Addressing the role of network influence in policy evaluation requires stakeholders to adopt rigorous experimental methods executed with nuanced sophistication \cite{athey2017}.  Some studies use randomization inference to estimate the treatment effect in the presence of network interference \cite{bowers2013reasoning,aronow2012general}, while others construct statistical estimators to estimate the policy effects \cite{leung2020treatment,tchetgen2012causal,forastiere2021identification}. There is also a stream of work that uses the linear-in-means model \cite{aronow2017estimating,cai2015}, which estimates the average treatment effect based on the aggregated individuals' covariates \cite{manski1993,kline2014}.

\looseness=-1
We note that most policy studies make use of fine-grained information, such as individuals' beliefs and demographic covariates, to aid in the estimation of the policy effects. However, due to certain constraints, policymakers do not necessarily have complete access to such information. The incomplete data, on top of the possible interaction between the selected policy-exposed respondents, inhibits the performance of these models and statistical estimators, thus creating a discrepancy in the estimates. Policymakers can account for this discrepancy by computing the distance between the initial and empirically observed belief distributions. Popular distance metrics include Wasserstein and total variation or divergences such as Kullback-Leibler (KL) and Jensen-Shannon (JS) \cite{cai2022distances}.  While metrics such as KL and JS divergences are often used in comparing probability distributions due to their ease of computation, they present several limitations. KL divergence is sensitive to small fluctuations in data samples and is agnostic to the geometry of the underlying distribution \cite{ozair2019wasserstein}. Moreover, both KL and JS divergences can become ill-defined when the distributions have non-overlapping support. In particular, the JS divergence fails to provide meaningful analysis in such settings \cite{kolouri2018sliced, arjovsky2017wasserstein}. These shortcomings make such metrics less amenable to analysis and statistical inference. In our work, we measure differences in distributions using the Wasserstein distance.  The Wasserstein distance is a special case of the optimal transport problem, which has the interpretation of being the minimum effort needed to shift mass from one distribution to another \cite{villani2009optimal}.  In particular, the Wasserstein distance provides a meaningful measure of difference between distributions even in settings where their support do not overlap.  Notably, it has various applications in machine learning \cite{frogner2015learning,leo2023wasserstein}, economics \cite{gini1914di} and finance \cite{rachev1998mass}. By quantifying the distance between the initial and empirically observed belief distributions using the Wasserstein metric, we highlight the effects of interaction and sample size and illustrate their tradeoff analytically and numerically.
\section{Problem Formulation}\label{sec:problemformulation}
Suppose a policymaker wishes to implement a specific policy on the population denoted by $P$.  We model the internal (private) response of each respondent in $P$ as a scalar random variable $X_i$ drawn independently from an {\em unknown} distribution $\mathcal{F}^\star$
$$X_i \sim \mathcal{F}^\star, \quad X_i \in \mathbb{R}.$$

We let $\mu$ and $\sigma$ denote the mean and standard deviation of $\mathcal{F}^\star$ respectively.\footnote{As a note, several opinion dynamics models such as the Deffuant-Weisbuch (DW) and the Hegselmann-Krause (HK) models have a bounded support interval of $[0,1]$ -- here, we do not make such assumptions.} Policymakers are ultimately interested in understanding the population's receptiveness towards a policy. We model this process as trying to estimate the unknown distribution $\mathcal{F}^\star$.  The policymaker does so by drawing some observations from $\mathcal{F}^\star$, for instance, by performing a survey.  Unfortunately, this process is complicated by the fact that individuals within the population communicate among themselves and, in the process, influence each other's opinions of the policy.  While it is in principle possible for the policymaker to survey every respondent and track their responses periodically so as to understand how their opinions form and alter, this approach is simply too expensive to be practical.  Instead -- and realistically speaking -- policymakers are only able to conduct a survey of the population's opinion in a single time instance.  Moreover, surveys are usually difficult to design and expensive to set up. By the time a survey is conducted, a substantial amount of time would have lapsed for the population to interact and reform opinions. 

We model relationships within the population as a graph $G(P,E)$, where the set of vertices is formed by the population $P$, and where the edges $E$ represent the relationship between pairs of respondents.  We distinguish between two scenarios: (i) where the relationships are mutual, in which case we model these relationships as an undirected graph $G$, and (ii) where relationships may be directed (e.g., a regular respondent who views social media posts by an influential person), in which case we model these relationships as a directed graph $G^d$. In the undirected case, we let $d_i$ denote the degree of respondent $i$ in $G$, i.e., it is the number of connections $i$ has.  Similarly, in the directed case, we let $d_i^o$ denote the outdegree of $i$, i.e., it is the number of respondents that $i$ knows. 

In this paper, we assume that policymakers have the ability to select the subset $S \subseteq P$ of the population $P$ to which the policy is revealed to.  There is an advantage in restricting access.  For instance, if we are able to reveal the policy only to a subset of the population who do not know each other, we minimize the possibility that these people interact and affect each other's opinions.  For simplicity, we also take $S$ to be the subset of the population that the policymakers eventually survey from in our subsequent discussions.  We let $\hat{\mathcal{F}}_{|S|}$ denote the empirical distribution of the response formed by $S$. 

Next, we model how people influence each other's opinions.  Suppose that policymakers reveal the policy to $n$ respondents, where $n < |P|$.  The interaction is represented by a stochastic matrix $A \in \mathbb{R}^{n \times n}$, where the entries $a_{ij}$ are non-negative and satisfy $\sum_j a_{ij} = 1$ for all $i \in [n]$.  The entries $a_{ij}$ model the weight respondent $i$ places on the opinion of respondent $j$ \cite{degroot1974reaching}.  After interaction, respondent $i$ updates its own opinion with $\sum_j
a_{ij} X_j$. In particular, the only entries $a_{ij}$ that can be non-zero are those where vertex $j$ is adjacent to $i$, i.e., $j \sim i$, or where $j = i$.
In the first part of this paper, we adopt the model where each respondent listens to each of his/her neighbor, including his/her own, equally, that is, 
$$
a_{ij} = \begin{cases}
   \frac{1}{|N_i|+1} \text{ for } j \in \{i\} \cup N_i, \\
   0 \text{ otherwise}.
\end{cases}
$$
Here, $N_i$ is the set of neighbors of respondent $i$.  We call this process the \emph{average interaction rule}. 
This rule is a simple mechanism to capture the effect of interactions.  There are settings where individuals may, for instance, place more emphasis on their personal opinions or may listen to certain people more than others.  We discuss these extensions where the opinions are formed by weighing the opinions of neighbors differently, i.e., the weighted interaction rule, in Section \ref{sec:extendedanalysis}. 

We summarize the sequence of events in Figure \ref{fig:Flowchart}.  First, policymakers select the subset of respondents $S$ to which the policy is revealed to.  The policy is revealed to $S$, and these respondents form their initial beliefs.  Next, they interact among themselves.  After interaction, they update their beliefs simultaneously.  In the last stage, policymakers observe their updated beliefs.


\begin{figure}[H]
\centering
\begin{tikzpicture}[node distance=1.3cm]

\node (sample) [process] {Sample respondents};
\node (expose) [process, below of=sample] {Expose policy};
\node (interact) [process, below of=expose] {Interact among respondents};
\node (update) [process, below of=interact] {Update beliefs};
\node (estimate) [process, below of=update] {Measure updated belief};

\draw [arrow] (sample) -- (expose);
\draw [arrow] (expose) -- (interact);
\draw [arrow] (interact) -- (update);
\draw [arrow] (update) -- (estimate);

\end{tikzpicture}
\caption{Flow chart of events}
\label{fig:Flowchart}
\end{figure}

\subsection{Sampling Strategies}
In this section, we describe how a subset of respondents is selected.  We describe different sampling strategies and discuss their implications in the estimation process. 

\textbf{Independent Set Sampling} \cite{ang2025estimating,karwa2018,fatemi2023network}. In this setting, the policy is revealed to a subset of $n$ respondents who are carefully chosen such that these respondents do not know one another. The group of respondents form an independent set in the relationship graph -- every pair of respondents does not communicate with one another. In particular, their beliefs remain unchanged between the policy revelation and the subsequent conduct of the survey.  Policymakers are then able to estimate the initial belief distribution provided they have surveyed sufficiently many respondents.

Suppose the samples being drawn i.i.d. from the same underlying $\mathcal{F}^\star$, then the Strong Law of Large Numbers (SLLN) tells us that
$$\hat{F}_n(t) \xrightarrow{a.s.} F(t) \text{ as } n \rightarrow \infty,$$
\looseness=-1
for every value of $t$, where $\hat{F}_n$ is the empirical cumulative distribution function (ecdf) of $\hat{\mathcal{F}}_n$ and $F$ is the cumulative distribution function (cdf) of $\mathcal{F}^\star$.  In other words, one is able to accurately estimate -- and without any {\em bias} -- the population's sentiment provided the sample size is sufficiently large.

\looseness=-1
In general, and as in all estimation problems, having as many samples as our budget allows is preferred.  There are two drawbacks to the independent set sampling strategy.  First, and most crucially, it assumes that policymakers have knowledge about the graph $G$, that is, they know the relationships among its population.  This is an extremely strong assumption and can be unrealistic in certain policy contexts. The second downside is that requiring the set of respondents to form an independent set (in the relationship graph $G$) places a hard limit on the number of samples one can draw, and this depends on the graphical structure.  If the average person knows a substantial fraction of the population, i.e., the average degree of the graph $G$ is high, then we are only able to select a small number of respondents while ensuring that none of them know each other.  For instance, the following result by Kwok \cite{west1996introduction} describes an upper bound on the size of the independent set based on the maximum degree of the network.
\begin{theorem}\label{thm: upperboundindpset}
Let $G(V,E)$ be a graph and let $\Delta$ be the maximum degree in $G$. Then, the size of the largest independent set of $G$, $\alpha(G)$, has the following upper bound,
$$ \textstyle \alpha(G) \leq |V| - \frac{|E|}{\Delta}.$$
\end{theorem}

\textbf{Cluster/Clique Sampling} \cite{ugander2013,ugander2020,eckles2017design}. 
In the previous set-up, selected respondents do not know each other, and hence there was no element of interaction.  Our next sampling strategy introduces a stylized model where we control for the degree of interaction.  
Concretely, we assume that the policymakers sample $p$ disjoint cliques $C_1, C_2, \ldots, C_p$,  each of size $r$.  As such, the total number of selected respondents is $n = pr$.  We let $\Tilde{G}$ denote the resulting subgraph describing the relationship between these sampled respondents
\[
\begin{cases}
e_{ij} \in \Tilde{E} & \text{ if } i,j \in C_k \text{ for } k \in \{1, 2, \ldots, p\},\\
e_{ij} \not\in \Tilde{E} & \text{ if } i \in C_k, \, j \in C_l \text{ for } k \neq l,
\end{cases}
\]
with $\Tilde{E}$ denoting the set of edges of $\Tilde{G}$. Respondents within the same clique interact with each other and influence each other's opinions.  Respondents from different cliques, on the other hand, do not interact. The purpose of introducing the cluster/clique sampling strategy is to help us understand how certain attributes about communities (such as the size of the communities $r$ as well as the number of different communities $p$) affect our estimation of $\mathcal{F}^\star$. 

\textbf{Random Sampling}. 
In the independent set sampling strategy, we assume perfect knowledge of the underlying relationship network $G$ and we sample respondents based on knowledge of $G$.  The random sampling strategy does the exact opposite -- we simply sample from a large pool without regard for the underlying relationship between respondents.  In doing so, we are potentially able to sample from a larger pool of respondents.  The downside is that there may be a fraction of respondents who know each other, that is, they share edges in $G$.  The policymakers have no control over this process.  These respondents may communicate about the policies among themselves before we are able to survey their response, and in the process, update their internal response to the policies. 

After illustrating the sampling strategies that the policymakers use, we explain how policymakers can account for the discrepancies between the observed belief distribution $\hat{\mathcal{F}}_n$ and the initial belief distribution $\mathcal{F}^\star$ in the next section.

\section{Quantifying the Distance between $\hat{\mathcal{F}}_n$ and $\mathcal{F}^\star$}\label{sec:mainresults}
\looseness=-1
The main source of information policymakers have to infer the response of the entire population is the surveyed responses from the sampled population. 
Based on these responses, the policymakers form an estimate of the underlying distribution of the population's private (unknown) response, which in turn forms the basis of their evaluation of the policies.  As such, it is of interest to understand how the estimated distribution differs from the true underlying distribution.

Specifically, policymakers collect surveyed responses $\{ X_i\}_{i=1}^{n}$.  Using these responses, the policymaker obtains an empirical distribution $\hat{\mathcal{F}}_n$ using the empirical data, which then serves as a comparison to the underlying unknown distribution $\mathcal{F}^\star$.  Both $\hat{\mathcal{F}}_n$ and $\mathcal{F}^\star$ are probability distributions, so a natural way to compare these is to deploy a suitable distance measure over probability distributions. 

\subsection{Optimal Transport Distance}

In this paper, we quantify the difference between probability distributions using the Wasserstein distance. The Wasserstein distance is a special instance of the optimal transport (OT) problem, which seeks optimal transportation plans between probability distributions so as to minimize cost.  Concretely, let $\alpha \in \Delta_m$ and $\beta \in \Delta_n$ be probability distributions over a metric space, here, $\Delta_m = \{\alpha \in \R^m_+, \sum^m_{i=1}\alpha_i = 1\}$ denotes the probability simplex. Let $C \in \R^{m \times n}$ be the matrix such that $C_{i,j}$ models the transportation cost between points $x_i \sim \alpha$ and $y_j \sim \beta$. The Wasserstein distance is defined as the solution of the following convex optimization instance
\begin{equation}\label{eqn: wasserstein}
\pi_W = \underset{\pi}{\arg\min}~ \langle C, \pi \rangle \qquad \mathrm{s.t.} \qquad \pi \in \Pi(\alpha, \beta)
\end{equation}
where $\Pi(\alpha, \beta) = \{\pi \in \R^{m \times n}_+: \pi\mathbf{1}_n = \alpha, \pi^T\mathbf{1}_m = \beta\}$ denotes the set of couplings between probability distributions $\alpha \in \Sigma_m, \beta \in \Sigma_n$ and $\mathbf{1}_m \in \R^m$ denotes the vector of ones. 

The OT problem \eqref{eqn: wasserstein} is an instance of a linear program (LP), and hence admits a global minimizer. 
In the special case where the points $x_i$ and $y_j$ lie on the real line $\R$, and where the cost between two points $x$ and $y$ is given by the absolute value $|x-y|^p$, the $W_p$ distance admits a simpler expression:
$$ \textstyle W_p(\xi,\nu) = \left(\int^1_0|G_1^{-1}(q)-G_2^{-1}(q)|^pdq\right)^{1/p},$$
where $\xi$ and $\nu$ are probability measures on $\R$, and $G_1^{-1}$ and $G_2^{-1}$ are their respective inverse cdf.
In particular, when $p = 1$, the $W_1$ distance simplifies to the following:
$$ \textstyle W_1(\xi,\nu) = \int_{\R}|G_1(x)-G_2(x)|dx,$$
where $\xi$ and $\nu$ are probability measures on $\R$, and $G_1(x)$ and $G_2(x)$ are their respective cdfs.



Suppose that the policymakers sample $n$ respondents. Given $n$ i.i.d. samples $X_1, X_2, \ldots, X_n$ from the distribution $\mathcal{F}^\star$, the ecdf of their beliefs $\hat{F}_n$, is given as 
$$\textstyle \hat{F}_n(t) = \frac{1}{n}\sum^n_{i=1}\mathds{1}(X_i \leq t).$$
The ecdf is a random measure.  As such, to quantify how $\hat{F}_n$ differs from $F$, we compute the expected $W_1$ distance between $\hat{F}_n$ and $F$.  Next, we investigate how the various sampling strategies change the ecdf and its implications on the $W_1$ distance.

\subsection{Independent Set Sampling}
As the independent set sampling strategy does not allow interactions between the sampled respondents, the observed belief distribution $\hat{\mathcal{F}}_n$ is just the empirical version of the initial belief distribution $\mathcal{F}^\star$. Then, the ecdf $\hat{F}_n(t)$ is given as $\frac{1}{n}\sum^n_{i=1}\mathds{1}(X_i \leq t)$. We note that, for a fixed $t$, $\mathds{1}(X_i \leq t)$ is a Bernoulli random variable, where 
\[
\mathds{1}(X_i \leq t) =
\begin{cases}
1 & \text{w.p. $F(t)$},\\
0 & \text{w.p. $1-F(t)$}. 
\end{cases}
\]
Then, $\E{\hat{F}_n(t)} = F(t)$ and $Var(\hat{F}_n(t)) = \frac{F(t)(1-F(t))}{n}$.


\begin{proposition}\label{prop: indpsample}
The expected $W_1$ distance between the ecdf $\hat{F}_n$ and the cdf $F$ has the following upper bound,
$$\textstyle \E{W_1(\hat{F}_n,F)} \leq \frac{1}{\sqrt{n}}\int_{\R}\sqrt{F(t)(1-F(t))}dt.$$
\end{proposition}
\begin{proof}
\begin{align*}
\textstyle \E{W_1(\hat{F}_n,F)} &= \textstyle \E{\int_{\R}|\hat{F}_n(t)-F(t)|dt} \\
& \textstyle \overset{(a)}{=} \int_{\R}\E{|\hat{F}_n(t)-F(t)|}dt \\
& \textstyle = \int_{\R}\E{\sqrt{(\hat{F}_n(t)-F(t))^2}}dt \\
& \textstyle \overset{(b)}{\leq}\int_{\R}\sqrt{\E{(\hat{F}_n(t)-F(t))^2}}dt  \\
& \textstyle = \frac{1}{\sqrt{n}}\int_{\R}\sqrt{F(t)(1-F(t))}dt 
\end{align*}
Here, (a) is by Fubini's Theorem and (b) is by Jensen's inequality on a concave square-root function.
\end{proof}

\looseness-1
In particular, Proposition \ref{prop: indpsample} tells us that the deviation between the initial and the empirically observed belief distributions vanishes at a rate of $O(1/\sqrt{n})$, where $n$ is the number of sampled respondents.  However, as we noted earlier, there is a hard limit to the value of $n$ to which we can apply the result.

\subsection{Clique Sampling}
Next, we consider the clique sampling strategy where policymakers obtain a sample of $n$ respondents by selecting $p$ cliques of size $r$ randomly. In this model, all respondents within the same clique know each other and interact among themselves.  On the other hand, pairs of respondents belonging to different cliques do not interact with each other.  Following the average interaction rule, the updated beliefs of each respondent after interacting are a random variable whose distribution is equal to the sample mean distribution of $r$ i.i.d. samples drawn from the distribution $\mathcal{F}^\star$.  In the following, we denote the sample mean distribution obtained by drawing $r$ i.i.d. samples from $\mathcal{F}^\star$ by $\mathcal{F}^r$, and we denote the cdf of $\mathcal{F}^r$ by $F^r$. Since all respondents within the same clique perform the same update, all of them share the same updated belief after interacting. Subsequently, the empirical distribution formed by the surveyed responses is equal to $p$ i.i.d. random variables drawn from $\mathcal{F}^r$, and each observation is repeated with $r$ copies.  We denote the resulting ecdf by $\hat{F}^r_p$.  Suppose we let $\mu$ and $\sigma$ denote the mean and the standard deviation of $\mathcal{F}^\star$.  Then, the mean and standard deviation of $\mathcal{F}^r$ is $\mu$ and $\sigma/\sqrt{r}$ respectively. Since the $W_1$ distance defines a metric, one has
$$
\mathbb{E}[W_1 (\hat{F}_p^r,F)] \leq \mathbb{E}[W_1 (\hat{F}_p^r,F^r)] + \mathbb{E}[W_1 (F^r,F)].
$$
In the following, we bound the term $\mathbb{E}[W_1 (\hat{F}_p^r,F)]$ by analyzing the two error terms on the RHS.  In particular, Proposition \ref{prop: indpsample} allows us to bound $\mathbb{E}[W_1 (\hat{F}_p^r,F^r)]$ directly.  As such, the next step is to understand how the sample mean distribution deviates from the original distribution, as a function of $r$.  We are not aware if a simple expression that bounds $\mathbb{E}[W_1 (F^r,F)]$ for general distributions exists.  However, explicit expressions for the 2-Wasserstein ($W_2$) distances between normal distributions are well known. As such, in the following, we establish a bound for $\mathbb{E}[W_1 (F^r,F)]$ by approximating $F^r$ and $F$ with a normal distribution with matching mean and variance.

First, we state a result that provides the $W_2$ distance between two distributions in terms of its cdfs.

\begin{theorem}[Equation 21 \cite{irpino2015basic}]\label{thm:2wasserstein}
Let $A$ and $B$ be two distributions.  Let $F_{A}$ and $F_{b}$ be the corresponding cdfs.  The $W_2$ distance between these distributions is given by 
\begin{equation*}
\resizebox{\hsize}{!}{$W^2_2(F_A, F_B) = (\mu_A-\mu_B)^2 + (\sigma_A-\sigma_B)^2 + 2\sigma_A\sigma_B(1-\rho^{A,B}),$}
\end{equation*}
\looseness-1
where $\mu_A$ and $\mu_B$ are the respective means, $\sigma_A$ and $\sigma_B$ are the respective standard deviations, and $\rho^{A,B}$ is the Pearson correlation of the points in the quantile-quantile plot of $F_A$ and $F_B$.
\end{theorem}

Next, let $\Tilde{F}$ and $\Tilde{F}^r$ be respective cdfs of the normal distributions, $N(\mu, \sigma^2)$ and $N(\mu, \frac{\sigma^2}{r})$.  Using these distributions, we obtain the following upper bound of the expected $W_1$ distance between $\hat{F}^r_p$ and $F$. 

\begin{proposition}\label{prop: clustersampling}
Suppose the policymakers select $n$ respondents such that there are $p$ cliques of size $r$, where $n = pr$. The expected $W_1$ distance between cdf of the initial belief distribution $F$ and the ecdf of the observed belief distribution $\hat{F}^r_p$ is given as
\begin{equation*}
\resizebox{\hsize}{!}{$
\begin{aligned}
\textstyle \E{W_1(F, \hat{F}^r_p)} \leq & \textstyle \; \sigma \big( 1 - \frac{1}{\sqrt{r}} \big) + \frac{1}{\sqrt{p}}\int_\R\sqrt{F^r(t)(1-F^r(t)}dt \\
& \textstyle + \sigma \sqrt{2 (1-\rho^{F,\Tilde{F}})} + \sigma \sqrt{(2/r)(1-\rho^{F^r,\Tilde{F}^r})},
\end{aligned}
$}
\end{equation*}
where $\rho^{F,\Tilde{F}}$ and $\rho^{F^r,\Tilde{F}^r}$ are the respective Pearson correlation of the points in the quantile-quantile plot of $F$ and $\Tilde{F}$, and $F^r$ and $\Tilde{F}^r$.
\end{proposition}

\begin{proof}
We note that there is an analytic solution of the $W_2$ distance between two Gaussian distributions, $\mathcal{G} = N(\mu_G, \sigma_G^2)$ and $\mathcal{H} = N(\mu_H, \sigma_H^2)$, which is given by
\begin{equation}\label{eqn:2wassersteinnormal}
W_2^2(\mathcal{G}, \mathcal{H}) = (\mu_G - \mu_H)^2 + (\sigma_G - \sigma_H)^2
\end{equation}
Then, we have
\begin{equation*}
\resizebox{.9\hsize}{!}{$
\begin{aligned}
&\E{W_1(F, \hat{F}^r_p)}\\ &\overset{(a)}{\leq} \E{W_1(F, F^r) + W_1(F^r, \hat{F}^r_p)} \\
&\overset{(b)}{\leq} W_1(F, \Tilde{F}) + W_1(\Tilde{F}, \Tilde{F}^r) + W_1(\Tilde{F}^r, \Tilde{F}) + \E{W_1(F^r, \hat{F}^r_p)}\\
&\overset{(c)}{\leq} W_2(F, \Tilde{F}) + W_2( \Tilde{F}, \Tilde{F}^r) + W_2(\Tilde{F}^r, \Tilde{F}) + \E{W_1(F^r, \hat{F}^r_p)}\\
& \textstyle \overset{(d)}{\leq} \sqrt{2\sigma^2(1-\rho^{F,\Tilde{F}})} + \sigma -  \frac{\sigma}{\sqrt{r}}\\
& \textstyle \quad + \sqrt{\frac{2\sigma^2}{r}(1-\rho^{F^r,\Tilde{F}^r})} + \E{W_1(F^r, \hat{F}^r_p)}\\
& \textstyle \overset{(e)}{\leq} \sqrt{2\sigma^2(1-\rho^{F,\Tilde{F}})} + \sqrt{\frac{2\sigma^2}{r}(1-\rho^{F^r,\Tilde{F}^r})} + \sigma -  \frac{\sigma}{\sqrt{r}} \\
& \textstyle \quad + \frac{1}{\sqrt{p}}\int_\R\sqrt{F^r(t)(1-F^r(t)}dt\\
& \textstyle = \sigma \big( 1 - \frac{1}{\sqrt{r}} \big) + \frac{1}{\sqrt{p}}\int_\R\sqrt{F^r(t)(1-F^r(t)}dt  \\
& \textstyle \quad + \sigma \sqrt{2 (1-\rho^{F,\Tilde{F}})} + \sigma \sqrt{(2/r)(1-\rho^{F^r,\Tilde{F}^r})}
\end{aligned}
$}
\end{equation*}
Here, (a) and (b) are by triangle inequality, (c) is from the fact that $W_1 \leq W_2$, (d) is by Theorem \ref{thm:2wasserstein} and Equation \ref{eqn:2wassersteinnormal}, and (e) is by Proposition \ref{prop: indpsample}. We drop the expectation at (b) since $F$, $F^r$, $\Tilde{F}$ and $\Tilde{F}^r$ are not random functions. 
\end{proof}

We explain the interpretation behind these bounds.  First, the term $\sigma ( 1 - \frac{1}{\sqrt{r}} )$ captures the effect of cliques, namely, each respondent pays more attention to the opinions of others and de-emphasizes his or her own opinions.  When $r=1$, the term is zero, and there is no deviation because each respondent does not communicate.  When $r$ is large, the error term is approximately $\sigma$, representing the other extreme where we get no information about the differences between opinions, as we only observe a consensus among the respondents.  The second term $\frac{1}{\sqrt{p}}\int_\R\sqrt{F^r(t)(1-F^r(t)}dt$ is analogous to Proposition \ref{prop: indpsample}. In particular, more samples $p$ lead to a more effective estimation of $F$.  For fixed $n$, the terms $p$ and $r$ represent an inherent tension, that is, larger cliques that form consensus mean that individuals lose their voice, and the effect is quantified in the first row of the error terms. 

We briefly comment on the error term $\sigma \sqrt{2 (1-\rho^{F,\Tilde{F}})} + \sigma \sqrt{(2/r)(1-\rho^{F^r,\Tilde{F}^r})}$ in the second row.  These represent the error when approximating the distributions with a normal distribution.  If the original distributions are suitably close to being normal, then the terms $(1-\rho^{F,\Tilde{F}})$ and $(1-\rho^{F^r,\Tilde{F}^r})$ are approximately zero, and we can ignore the contributions of these terms.  We believe the presence of these terms is an artefact of our analysis. In particular, even if the original distribution $\mathcal{F}^\star$ is far from a normal distribution, we believe that the true error $\E{W_1(F, \hat{F}^r_p)}$ should be dominated by the terms in the first row.

\subsection{Random Sampling}\label{subsec:randomsampling}
Lastly, we consider the random sampling strategy, where the policymaker selects $n$ respondents from the population uniformly at random. The analysis in this setting is the most challenging because the resulting subgraph $\Tilde{G}$ induced by the selected respondents is random. Therefore, the number of neighbors each respondent has (i.e., his/her degree) determines the amount of interaction within $\Tilde{G}$.  Suppose that the policymakers select $n$ respondents uniformly at random, and we assume that the sampled respondents follow the average interaction rule. Let the initial beliefs of each of the respondents be a random variable denoted by $X_1, X_2, \ldots, X_n$, which are drawn i.i.d. from $\mathcal{F}^\star$. We denote the random variable of the updated belief of each sampled respondent $i$ to be $X'_i := \frac{1}{|N_i|+1}(X_i + \sum_{j \in N_i} X_j)$ and the ecdf of the observed belief distribution is given by
\begin{equation*}
\textstyle \hat{F}_n(t) = \frac{1}{n}\sum_{i=1}^n\mathds{1}\left(X'_i \leq t\right).
\end{equation*}
We note that the random variable $X'_i$ follows a sample mean distribution of $\mathcal{F}^\star$ with size $d_i + 1$, where $d_i:=|N_i|$ is the degree of respondent $i$.  We denote the resulting distribution $\mathcal{F}^{d_i+1}$, and we denote the corresponding cdf as $F^{d_i+1}$, for all $i \in \{1,2,\ldots,n\}$.  Then, the mean and standard deviation of $\mathcal{F}^{d_i+1}$ is $\mu$ and $\sigma/(\sqrt{d_i+1})$ respectively. We let $Y_i(t)$ be the indicator variable $\mathds{1}(X'_i \leq t)$. For a fixed $t$, $Y_i(t)$ is a Bernoulli random variable, where
\[
Y_i(t) =
\begin{cases}
1 & \text{w.p. $F^{d_i+1}(t)$},\\
0 & \text{w.p. $1-F^{d_i+1}(t)$}.
\end{cases}
\]
\looseness=-1
Then, $\E{Y_i(t)} = F^{d_i+1}(t)$ and $Var(Y_i(t)) = F^{d_i+1}(t)(1-F^{d_i+1}(t))$. Before we derive the upper bound of the expected $W_1$ between the ecdf of the observed belief distribution $\hat{F}_n$ and the cdf of the initial belief distribution $F$, we first consider a mixture distribution over the set of cdfs, $F^{d_1+1}, F^{d_2+1}, \ldots, F^{d_n+1}$ with equal weights, where we denote the cdf $F_n$ as $\frac{1}{n}\sum_{i=1}^nF^{d_i+1}$. Then, $\E{\hat{F}_n(t)} = F_n(t)$.



Our analysis is complicated by the fact that interactions can introduce ``long range" correlations. 
Consider two respondents $A$ and $C$ who do not know each other, but share a common friend $B$ (neighbor).  Even though $A$ and $C$ do not interact with each other directly, their interactions with $B$ induce an indirect form of communication. As such, a key part of our analysis is to bound the effects of these ``long range" connections.  Specifically, we bound the number of respondents that are of distance at most two from a respondent, which shows that the influence of these ``long range" connections is limited.


In what follows, for every respondent $i$, we consider the set of neighbors, denoted by $N_i$, as well as the set of vertices that are of distance 2 (i.e., two hops) from $i$, denoted by $M_i$. To capture this correlation, we make use of a 2-star graph, as defined below, to quantify the number of vertices that are two hops away from a given vertex.

\begin{definition}
A 2-star graph has 3 vertices and 2 edges, where the central vertex is adjacent to 2 leaf vertices.
\end{definition}
Given the adjacency matrix $A$ of the resulting subgraph after sampling, the $(r,s)$-entry of $A^2$ gives the number of paths of distance 2 from respondent $r$ to $s$. Hence, using $A$ and $A^2$, we define the indicator variable $M_{r,s}$, where respondent $s$ is at most two hops from $r$, as follows,
$$
M_{r,s} =
\begin{cases}
1 & \text{if } A_{r,s} + A^2_{r,s} \geq 1,\\
0 & \text{otherwise},
\end{cases}$$
\looseness -1
where $A_{r,s}$ and $A^2_{r,s}$ are the $(r,s)$-entry of $A$ and $A^2$ respectively. We also make the following assumption to obtain an upper bound on the number of common neighbors that two sampled respondents share. This allows us to achieve a tighter bound, as illustrated in the deferred analysis. 
\begin{assumption}\label{assume:fewdegree}
Given a graph $G(V,E)$, we assume that there are at most $2n(\langle d \rangle + \langle d \rangle^2)$ ``long range" connections, i.e., $\sum_{r \neq s}M_{r,s} \leq 2n(\langle d \rangle + \langle d \rangle^2)$, where $\langle d \rangle$ is the average degree of $G$.
\end{assumption}

We substantiate the assumption using the Erd\"{o}s-Renyi (E-R) random graph model, $\mathcal{R}(|V|,p)$, in the following claim. 
\begin{claim}\label{claim:assumptionfewdegree}
Let $G(V,E)$ be a graph where the edges are formed according to the E-R random graph model $\mathcal{R}(|V|,p)$. Let $\mathbf{L}_v$ be the random variable for the number of vertices that are within two hops away from a given vertex $v \in V$. Then, for all vertices in $G$, Assumption \ref{assume:fewdegree} holds with probability 
$$\textstyle \Prob{\bigcap_{v \in V}(\mathbf{L}_v \leq 2 \langle d \rangle + 2 \langle d \rangle^2)} \gtrsim 1- \frac{2}{|V|p^2}.$$
\end{claim}
As a proof sketch, for a given vertex $v$, we first provide a probabilistic bound on $|N_v|$ by considering the degree of $v$. Then, we bound $|M_v|$ using 2-star graphs with $v$ as a leaf vertex. Lastly, we will invoke the union bound to prove the claim. The full proof is given in Appendix \ref{appendix: proofofclaim}.

Claim \ref{claim:assumptionfewdegree} shows that for a graph $G(V,E)$, there are at most $2n(\langle d \rangle + \langle d \rangle^2)$ ``long range" connections with high probability, and it suggests that Assumption \ref{assume:fewdegree} is reasonable.
This reflects that the effect of the ``long range" correlations is limited. While Claim \ref{claim:assumptionfewdegree} is established under the E-R graph, which lacks several features of real-world social networks, such as scale-free and ``small-world" properties, the analysis serves as a basis to demonstrate the strength of such influence analytically. Furthermore, as real-world social networks are sparse and have high local clustering, individuals tend to have a set number of neighbors and mutual connections. Hence, such characteristics make Assumption \ref{assume:fewdegree} a reasonable approximation in practical contexts. 

\looseness=-1
In what follows, we provide an upper bound of the expected $W_1$ distance between $\hat{F}_n$ and $F$, supposing Assumption \ref{assume:fewdegree} holds. We note that there are two elements of randomness, specifically, the sampling strategy and their initial beliefs, which are independent of each other. Hence, to quantify the expected $W_1$ distance, we invoke the law of iterated expectation. We first consider the ``inner" expectation after $n$ respondents have been selected uniformly at random and state the following proposition for the upper bound. In the following, we let $\Tilde{F}$ and $\Tilde{F}^{d_i+1}$ be the cdfs of the normal distributions $N(\mu, \sigma^2)$ and $N(\mu, \frac{\sigma^2}{d_i+1})$ for $i \in [n]$.

\begin{proposition}\label{prop:randomsampling}
\looseness=-1
Given the resulting subgraph $\Tilde{G}$ obtained by randomly selecting $n$ respondents.  Suppose Assumption \ref{assume:fewdegree} holds.  Then, the expected $W_1$ distance between the ecdf of the observed belief distribution $\hat{F}_n$ and the cdf of the initial belief distribution $F$ has the following upper bound, 
\begin{equation*}
\resizebox{\hsize}{!}{$
\begin{aligned}
\E{W_1(\hat{F}_n,F)}
& \textstyle \leq \frac{1}{n}\int_{\R}\sqrt{\sum_{i=1}^nF^{d_i+1}(t)(1-F^{d_i+1}(t))} dt\\
& \textstyle \quad + \sigma(1-\frac{1}{n}\sum^n_{i=1} \frac{1}{\sqrt{d_i+1}}) \\
& \textstyle \quad + \sigma\sqrt{2(1-\rho^{F,\Tilde{F}})} + \frac{\sigma}{n}\sum_{i=1}^n \sqrt{\frac{2(1-\rho^{F^{d_i+1},\Tilde{F}^{d_i+1}})}{d_i+1}}\\
& \textstyle \quad + \frac{O(\langle d \rangle)}{\sqrt{n}},
\end{aligned}
$}
\end{equation*}
\vspace{-5pt}

where $\rho^{F,\Tilde{F}}$ and $\rho^{F^{d_i+1},\Tilde{F}^{d_i+1}}$ are the respective Pearson correlation of the points in the quantile-quantile plot of $F$ and $\Tilde{F}$, and $F^{d_i+1}$ and $\Tilde{F}^{d_i+1}$ for $i \in [n]$ and $\langle d \rangle$ is the average degree of $\Tilde{G}$.
\end{proposition}

\begin{proof}
Using a similar proof technique in Proposition \ref{prop: clustersampling}, we first obtain the following inequality,
\vspace{2pt}
\begin{equation*}
\begin{aligned}
\E{W_1(\hat{F}_n,F)} &\overset{(a)}{\leq} \E{W_1(\hat{F}_n,F_n)} + \E{W_1(F_n, F)}  \\
& \textstyle \overset{(b)}{\leq} \int_{\R}\sqrt{Var(\hat{F}_n(t))} dt + W_1(F_n,F)
\end{aligned}
\end{equation*}
\vspace{-5pt}

where (a) is by triangle inequality, (b) is by Proposition \ref{prop: indpsample} as $\E{\hat{F}_n(t)} = F_n(t)$. Since $F_n$ and $F$ are not random functions, we drop the expectation in the second term. 

In the variance term, we note that there is a non-zero probability that the sampled respondents know each other or have common neighbors. Hence, there are covariance terms, and the indicator variable $M_{r,s}$ detects whether sampled respondents $r$ and $s$ share an edge or a common neighbor. For a fixed $t$,  we have,
\begin{equation*}
\resizebox{.9\hsize}{!}{$
\begin{aligned}
&Var(\hat{F}_n(t))\\
&= \textstyle \frac{1}{n^2}\left(\sum_{i=1}^nVar(Y_i(t)) + \sum_{r \neq s}Cov(Y_r(t),Y_s(t))\right)\\
& \textstyle \overset{(a)}{\leq} \frac{1}{n^2}(\sum_{i=1}^nF^{d_i+1}(t)(1-F^{d_i+1}(t)) \\
& \textstyle \quad + \sum_{r \neq s}M_{r,s}\sqrt{Var(Y_r(t))Var(Y_s(t))})\\
& \textstyle \overset{(b)}{\leq} \frac{1}{n^2}\left(\sum_{i=1}^nF^{d_i+1}(t)(1-F^{d_i+1}(t)) + C_t(\langle d \rangle + \langle d \rangle^2)n\right)
\end{aligned}
$}
\end{equation*}
\vspace{-5pt}

where (a) is by Cauchy-Schwarz inequality, (b) is by Assumption \ref{assume:fewdegree} and $C_t = \max_{i}\{F^{d_i+1}(t)(1-F^{d_i+1}(t))\}$. By substituting the variance term into the integral, we have
\vspace{2pt}
\begin{equation*}
\resizebox{\hsize}{!}{$
\begin{aligned}
&\textstyle \int_{\R}\sqrt{Var(\hat{F}_n(t))} dt  \\
&= \textstyle \frac{1}{n}\int_{\R}\sqrt{\sum_{i=1}^nF^{d_i+1}(t)(1-F^{d_i+1}(t)) + C_t(\langle d \rangle + \langle d \rangle^2)n} dt\\
& \textstyle \leq \frac{1}{n}\int_{\R}\sqrt{\sum_{i=1}^nF^{d_i+1}(t)(1-F^{d_i+1}(t))}dt \\
& \textstyle \quad + \frac{\sqrt{(1+\frac{1}{\langle d \rangle})\langle d \rangle^2n}}{n}\int_\R \sqrt{C_t} dt\\
& \textstyle = \frac{1}{n}\int_{\R}\sqrt{\sum_{i=1}^nF^{d_i+1}(t)(1-F^{d_i+1}(t))}dt + \frac{O(\langle d \rangle)}{\sqrt{n}}
\end{aligned}  
$}
\end{equation*}
\vspace{-5pt}

We note that splitting the square-root term may be unnecessary, but it helps to improve the ``visibility" of the contribution of the sample size term and the term representing the presence of the shared edges and common neighbors to the upper bound. For the second term, we use a similar technique in Proposition \ref{prop: clustersampling}, in which we obtain the following upper bound,

\begin{equation*}
\resizebox{\hsize}{!}{$
\begin{aligned}
& W_1(F_n,F) \\
& \textstyle = W_1\left(\frac{\sum_{i=1}^nF^{d_i+1}}{n}, \frac{\sum_{i=1}^nF}{n}\right)\\
& \textstyle \overset{(a)}{\leq} W_1\left(\frac{\sum_{i=1}^nF^{d_i+1}}{n}, \frac{\sum_{i=1}^n\Tilde{F}^{d_i+1}}{n}\right)  + W_1\left(\frac{\sum_{i=1}^n\Tilde{F}^{d_i+1}}{n}, \frac{\sum_{i=1}^n\Tilde{F}}{n}\right) \\
& \textstyle \quad + W_1\left(\frac{\sum_{i=1}^n\Tilde{F}}{n}, \frac{\sum_{i=1}^nF}{n}\right)\\
& \textstyle \overset{(b)}{\leq} \frac{1}{n}\sum_{i=1}^nW_1(F^{d_i+1}, \Tilde{F}^{d_i+1}) + W_1(\Tilde{F}^{d_i+1},\Tilde{F}) + W_1(\Tilde{F},F)\\
& \textstyle \overset{(c)}{\leq} \sigma\sqrt{2(1-\rho^{F,\Tilde{F}})} + \sigma (1-\frac{1}{n}\sum^n_{i=1} \frac{1}{\sqrt{d_i+1}}) \\
& \textstyle \quad + \frac{\sigma}{n}\sum_{i=1}^n \sqrt{\frac{2(1-\rho^{F^{d_i+1},\Tilde{F}^{d_i+1}})}{d_i+1}}
\end{aligned}
$}
\end{equation*}

Here, (a) and (b) are by triangle inequality, and (c) is using a similar technique in Proposition \ref{prop: clustersampling}. We then obtain the bound by combining the two terms.
\end{proof}

We highlight two effects from the bounds stated in Proposition \ref{prop:randomsampling}. The term in the first row can be interpreted as the \emph{sample size effect} while the other terms can be interpreted as the deviation in the policy belief, i.e., \emph{interaction effect}. Intuitively, the interaction effect can be attributed to the extent of interaction among respondents -- terms in the second and third row in RHS, and the number of respondents one interacts with -- term in the last row in RHS. Through this bound, we illustrate the tradeoff between the sample size effect and the interaction effect. Although the policymakers could weaken the sample size effect by increasing the sample size while working within their constraints, they would have to discern the estimation performance loss due to the possible interaction effect. Moreover, the number of neighbors that a respondent has also contributes to the shift in his or her belief. Respondents with a few neighbors may experience a larger deviation in their beliefs as they are influenced by fewer individuals, thus more sensitive to extreme opinions.

As Proposition \ref{prop:randomsampling} is based on a given resulting subgraph after sampling, the policymakers might not know the entire population network, and the topology of this network can affect the resulting subgraph through the sampling process. Since the resulting subgraph varies across different samples of $n$ respondents, this implies that the size of the sample mean distribution $\mathcal{F}^{d_i+1}$, for each respondent $i$, is dependent on the resulting subgraph. By conditioning on the set of sampled respondents, we fix the subgraph $\Tilde{G}$ and we apply the result from Proposition \ref{prop:randomsampling}. Hence, using the law of iterated expectation, we have the following corollary.

\begin{corollary}\label{cor:randomsamplinggeneral}
Given the population graph $G(P,E)$, the expected $W_1$ distance between the ecdf $\hat{F}_n$ and the cdf of the initial belief distribution $F$ is given as
$$\mathbb{E}[W_1(\hat{F}_n,F)] =  \E{\E{W_1(\hat{F}_n,F)|\Tilde{G}}},$$
where the inner expectation is from Proposition \ref{prop:randomsampling}, which we then take the expectation over all possible resulting subgraphs that consist of $n$ respondents. 
\end{corollary}

As a note, we provide Corollary \ref{cor:randomsamplinggeneral} to address the context of the problem. Explicitly stating the full expectation is rather convoluted, and it doesn't provide additional insights, so we provide the expansion in Appendix \ref{appendix: expansionofcor}. However, the bounds are concentrated around the expected degree of the sampled respondent in the resulting subgraph, which depends on population size $|V|$, sample size $n$ and the degree of each respondent in the population.

\subsection{Extended Analysis} \label{sec:extendedanalysis}
In the preceding sections, we model interactions using an undirected network, and we assume each respondent weighs the opinion of all other respondents that he or she knows, including his or her own, equally. In this subsection, we explore a series of extensions.  Specifically, we consider different topologies within the relationship network between respondents, and we consider alternative interaction models.

\textbf{Directed Networks}.
Directed network structures mimic the situation where there are influencers and followers in society, and the followers update their beliefs with reference to those of the influencers. Each respondent may simultaneously act as both a follower and an influencer. That is, the respondent updates his or her belief after listening to his or her influencers, while also serving as a source of influence for other respondents. Moreover, a respondent is not restricted to having a single influencer or a single follower; rather, he or she may listen to or influence multiple respondents within their directed neighborhood. This structure reflects real-world situations, where an individual could be both a follower and an influencer, for example, educators or social media users who simultaneously receive input from others and shape opinions within their networks. To effectively model this phenomenon, we consider the random sampling strategy. We note that the independent set sampling strategy is ineffective in modeling this phenomenon, as this strict sampling strategy neglects the directionality in the network. By assuming that the sampled respondents follow the average interaction rule, the update depends on the out-degree of each node, as denoted by $d^o$, i.e., the number of neighbors that each respondent listens to in the resulting subgraph. We derive several analogous results to Assumption \ref{assume:fewdegree}, Proposition \ref{prop:randomsampling} and Corollary \ref{cor:randomsamplinggeneral}.

\begin{assumption}\label{assume:fewdirecteddegree}
Given a directed graph $G^d(P,E^d)$, we assume that there are at most $ 2n(\langle d^o \rangle + \langle d^o \rangle^2)$ ``long range" directed connections, where $\langle d^o \rangle$ is the average out-degree in the population network $G^d$.
\end{assumption}

\begin{corollary}\label{cor: directednetwork}
Given the directed resulting subgraph $\Tilde{G}^d$ obtained by randomly selecting $n$ respondents. Suppose Assumption \ref{assume:fewdirecteddegree} holds. Then, the expected $W_1$ distance between the ecdf of the observed belief distribution $\hat{F}_n$ and the cdf of the initial belief distribution $F$ has the following upper bound,
\begin{equation*}
\resizebox{\hsize}{!}{$
\begin{aligned}
\textstyle \E{W_1(\hat{F}_n,F)}
& \textstyle \leq \frac{1}{n}\int_{\R}\sqrt{\sum_{i=1}^n\Tilde{F}_{d^o_i+1}(t)(1-\Tilde{F}_{d^o_i+1}(t))} dt\\
&\textstyle \quad + \sigma\sqrt{2(1-\rho^{F,\Tilde{F}})} + \sigma(1-\frac{1}{n}\sum_{i=1}^n\frac{1}{\sqrt{d^o_i+1}}) \\ 
& \textstyle \quad + \frac{\sigma}{n}\sum_{i=1}^n \sqrt{\frac{2(1-\rho^{F^{d^o_i+1},\Tilde{F}^{d^o_i+1}})}{d^o_i+1}}\\
&\textstyle \quad + \frac{O(\langle d^o \rangle)}{\sqrt{n}},
\end{aligned}
$}
\end{equation*}
where $\rho^{F,\Tilde{F}}$ and $\rho^{F^{d^o_i+1},\Tilde{F}^{d^o_i+1}}$ are the respective Pearson correlation of the points in the quantile-quantile plot of $F$ and $\Tilde{F}$, and $F^{d^o_i+1}$ and $\Tilde{F}^{d^o_i+1}$ for $i \in [n]$ and $\langle d^o \rangle$ is the average degree of $\Tilde{G}^d$.
\end{corollary}

\looseness-1
\begin{corollary}\label{cor:directednetworksamplinggeneral}
Given the directed population graph $G^d$, the expected $W_1$ distance between the ecdf $\hat{F}_n$ and the cdf of the initial belief distribution $F$ is given as
$$\mathbb{E}[W_1(\hat{F}_n,F)] = \E{\E{W_1(\hat{F}_n,F)|\Tilde{G}^d}},$$
where the inner expectation is from Corollary \ref{cor: directednetwork}, which we then take the expectation over all possible resulting directed subgraphs that consist of $n$ respondents. 
\end{corollary}

Although the main difference in the derivation of Corollary \ref{cor: directednetwork} and Proposition \ref{prop:randomsampling} lies in the directionality of the network, we are still able to distinguish the two main effects: the interaction and sample size effects in the respective terms of Corollary \ref{cor: directednetwork} from similar terms. As the bounds depend on the out-degree of every node, ``highly-influential" respondents (i.e., those with high out-degree) have a stronger ability to shape policy beliefs. On the other hand, respondents with low out-degree contribute less to the interaction effect. 

\textbf{Weighted Interaction Rule}. Our analysis thus far assumed that every respondent listens to each of their neighbor to the same extent and updates one’s belief by taking a simple average. However, respondents can have a preference for who they would listen to and update their beliefs differently. To model this, we allow the respondents to place different weights on his or her own and their neighbors’ beliefs. As such, the sampled respondents follow the weighted average interaction rule. Suppose that the policymakers adopt the random sampling strategy and sample $n$ respondents from the population uniformly at random. We denote $a_{ij}$ to be the normalized tendency for respondent $i$ to listen to respondent $j$ after they interact, while satisfying the condition that $\sum_j a_{ij} = 1$ for all $i \in \{1,2,\ldots,n\}$ and $\mathcal{C}_i$ to be weighted sample mean distribution based on $a_{ij}$, with $C_i$ being its cdf. We derive several analogous results to Proposition \ref{prop:randomsampling} and Corollary \ref{cor:randomsamplinggeneral}.

\begin{corollary}\label{cor: generalinteraction}
Given the resulting subgraph $\Tilde{G}$ obtained by randomly selecting $n$ respondents. Suppose Assumption \ref{assume:fewdegree} holds. Then, the expected $W_1$ distance between the ecdf of the observed belief distribution $\hat{F}_n$ and the cdf of the initial belief distribution $F$ has the following upper bound,
\begin{equation*}
\resizebox{\hsize}{!}{$
\begin{aligned}
\textstyle \E{W_1(\hat{F}_n,F)}
&\textstyle \leq \frac{1}{n}\int_{\R}\sqrt{\sum_{i=1}^nC_i(t)(1-C_i(t)) } dt\\
&\textstyle \quad + \sigma\sqrt{2(1-\rho^{F,\Tilde{F}})} + \sigma \\
& \textstyle \quad + \frac{\sigma}{n}\sum_{i=1}^n  \sqrt{\sum_{j \in N^d_i \cup i} a_{ij}^2}\left(\sqrt{2(1-\rho^{C_i,\Tilde{F}_i})} -1\right) \\
&\textstyle \quad + \frac{O(\langle d \rangle)}{\sqrt{n}},
\end{aligned}
$}
\end{equation*}
\looseness=-1
where $\rho^{F,\Tilde{F}}$ and $\rho^{C_i,\Tilde{F}_i}$ are the respective Pearson correlation of the points in the quantile-quantile plot of $F$ and $\Tilde{F}$, and $C_i$ and $\Tilde{F}_i$ for $i \in [n]$ and $\langle d \rangle$ is the average degree of $\Tilde{G}$.
\end{corollary}

\begin{corollary}
For a given population graph $G$, the expected $W_1$ distance between the ecdf $\hat{F}_n$ and the cdf of the initial belief distribution $F$ is given as
$$\mathbb{E}[W_1(\hat{F}_n,F)] =  \E{\E{W_1(\hat{F}_n,F)|\Tilde{G}}},$$
where the inner expectation is from Corollary \ref{cor: generalinteraction}, which we then take the expectation over all possible resulting subgraphs that consist of $n$ respondents. 
\end{corollary}

Although the main difference in the derivation of Corollary \ref{cor: generalinteraction} and Proposition \ref{prop:randomsampling} lies in the degree of interaction among the respondents, we are still able to distinguish the two main effects: the interaction and sample size effects in the respective terms of Corollary \ref{cor: directednetwork} from similar terms. 
Although these personalized weights $a_{ij}$ reflect a more nuanced dynamic, obtaining these values requires sophisticated empirical studies. Although several statistical and machine learning approaches can be employed to infer these weights (see Section \ref{subsec:opinion}), the resulting estimates may be sensitive to sampling bias and model misspecification. Further sensitivity analyses can be conducted to examine how variations in the modeled interaction patterns influence belief estimates, for instance, specifying plausible ranges for interaction weights to assess model uncertainty in the bounds. However, small perturbations in the weights are unlikely to cause large changes in the bound, and the fundamental tradeoff between sample size and interaction effects still holds.  

To summarize, Propositions \ref{prop: indpsample} to \ref{prop:randomsampling} provide quantitative upper bounds that highlight how social interactions may distort perceived policy beliefs, and how sample sizes and the interactions affect the estimation of the population beliefs. While the choice of sampling strategy is often constrained by practical considerations such as cost and feasibility, in many cases, some respondents likely know each other and interact among themselves, thereby causing a shift in beliefs. Complementary measures, such as effective and targeted policy messaging, can be adopted to reduce variance in the initial beliefs. In this way, respondents would share similar initial beliefs, thus reducing the deviation in belief post-interaction. As a result, policymakers can better estimate the public sentiment and the policy effectiveness.


\section{Numerical Experiments}\label{sec:numericalexperiments}
In this section, we conduct a series of numerical experiments to compute and analyze the $W_1$ distances over a range of sampling strategies, network topologies and degrees of interaction. These experiments aim to (i) validate the theoretical bounds established in the earlier analysis, and (ii) investigate how opinion among sampled respondents changes across various networks and interaction contexts.

\subsection{Validating Theoretical Bounds}
In this subsection, we compare the mean $W_1$ distance between the initial and updated belief distributions under different sampling strategies, to their corresponding theoretical bounds.  The purpose of these experiments is to investigate the strength of the bounds in the earlier analyses.

We examine two different types of belief distributions. First, we consider initial belief distributions with different variances. Specifically, we investigate distributions whose density is concentrated about a small interval. This mimics situations where most of the population shares similar opinions, while a few respondents take on extreme beliefs about a policy. We model the initial beliefs using two beta distributions, namely Beta(2, 2) and Beta(2, 5) distributions. 
Second, we consider initial belief distributions with unbounded support, for the sake of \emph{generality}, as our analysis allows the belief values to take up any real number. We use the standard normal distribution to model the initial beliefs.

\looseness=-1
For this analysis, we employ two sampling strategies, specifically the independent set sampling and cluster (clique) sampling strategies, as they represent two extremes: the independent set contains no edges, whereas cliques are densely connected with many edges. For the independent set sampling strategy, we sample up to 200 respondents, whereas for the cluster sampling strategy, we sample 200 respondents with varying clique sizes. To compute the empirical mean $W_1$ distance under different initial belief distributions and sampling strategies, we run the numerical experiments 500 times. In the following, we present the results for independent set and cluster sampling strategies in Figures \ref{fig:indpsetbound} and \ref{fig:clusterbound}, respectively. 

\begin{figure}
    \centering
    \includegraphics[width=\linewidth]{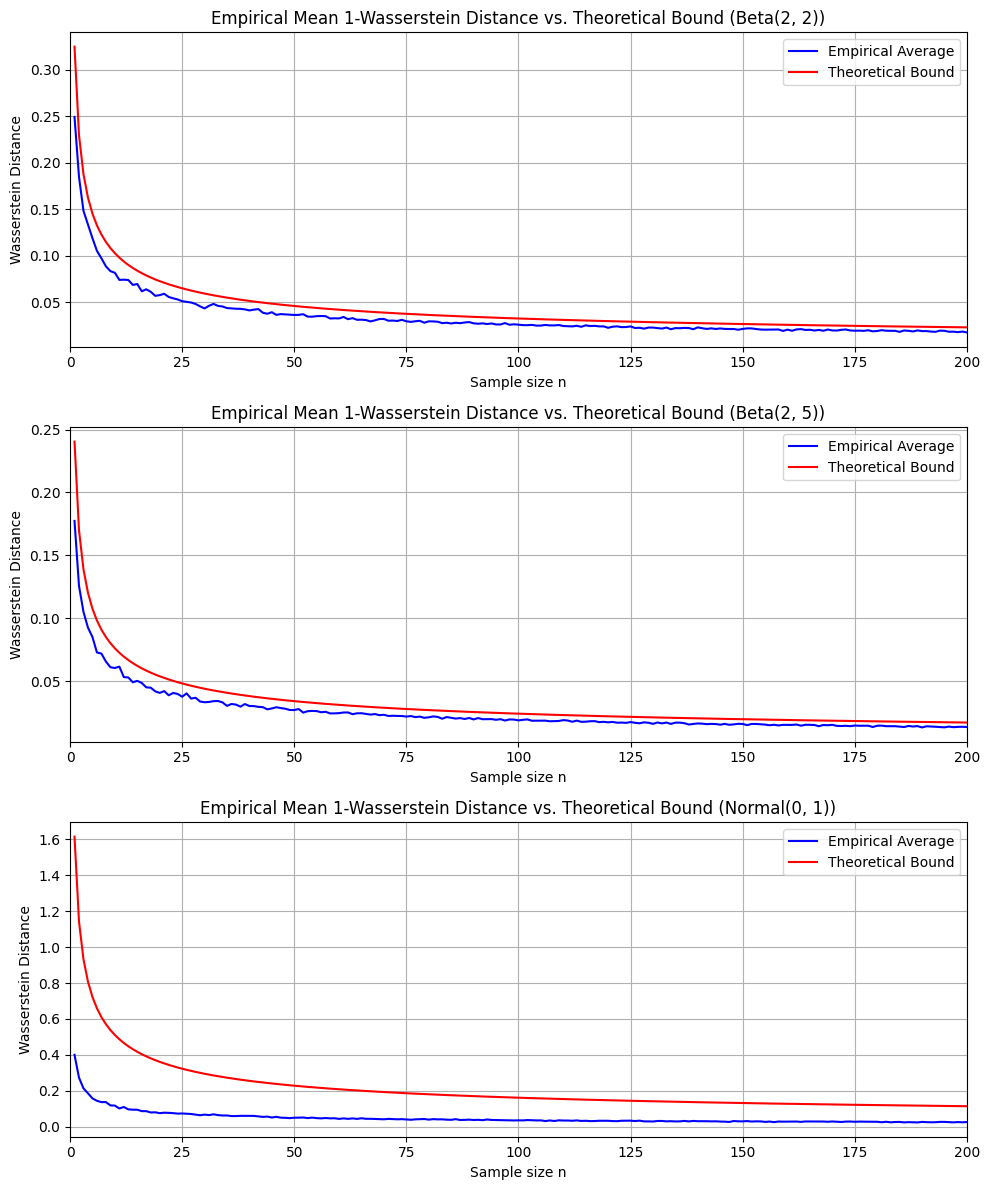}
    \caption{Theoretical bounds and empirical mean $W_1$ distances for independent set sampling strategy}
    \label{fig:indpsetbound}
    \vspace{-10pt}
\end{figure}
\looseness=-1
\textbf{Independent Set Sampling}. Figure \ref{fig:indpsetbound} validates the derived upper bound in Proposition \ref{prop: indpsample} as we observe both the theoretical upper bound and mean $W_1$ distance follow a similar decay rate. This highlights the fact that as the sample size increases, policymakers obtain more accurate estimates, and the mean $W_1$ distance decreases. We also observe that the theoretical upper bounds of the beta distributions seem to be less conservative than those of the standard normal distribution, especially at small sample sizes. This may be attributed to the bounded support, which limits the extent of the deviation in beliefs, resulting in a smaller mean $W_1$ distance.

\begin{figure}
    \centering
    \includegraphics[width=\linewidth]{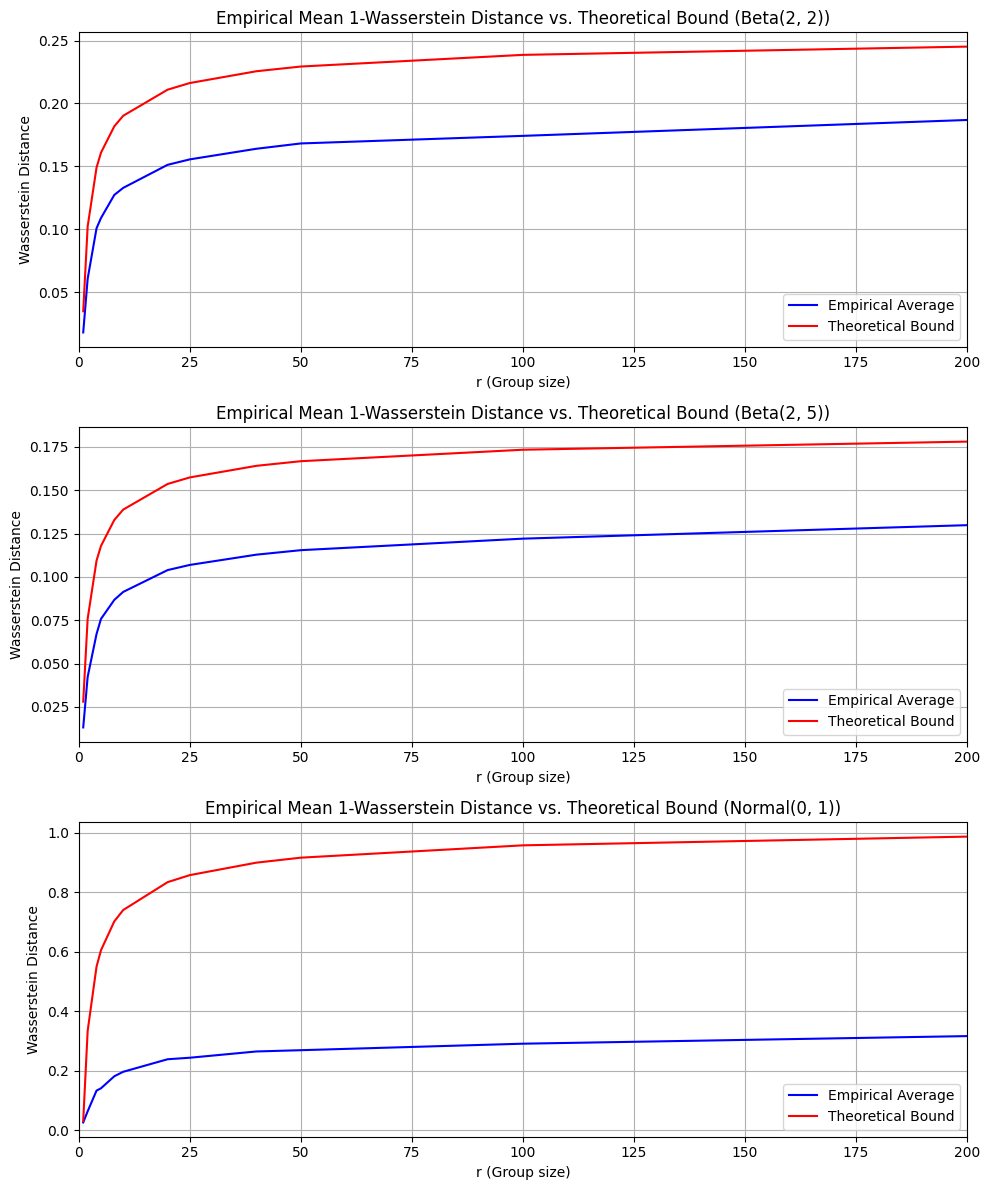}
    \caption{Smoothed theoretical bounds and empirical mean $W_1$ distances for cluster sampling strategy}
    \label{fig:clusterbound}
    \vspace{-10pt}
\end{figure}

\textbf{Cluster/Clique Sampling}. Unlike the independent set sampling strategy, the bounds are more conservative under the clique sampling strategy, especially at larger clique sizes. This is likely due to the stronger influence of interaction effects. These effects are inherently harder to quantify, which are reflected by the approximations made in the earlier theoretical analysis, thus leading to more conservative upper bounds. Nonetheless, the curves of the theoretical upper bound and the empirical mean $W_1$ distance have a similar structure, where they increase at a decreasing rate against the clique size.

\subsection{Experiments Using Synthetic and Real-World Networks}
Next, we investigate how different population networks and interaction mechanisms impact opinions.  We conduct our experiments over four different networks: (i) a synthetically generated population network based on the Erd\"{o}s-Renyi (E-R) model \cite{erdHos1960evolution}, (ii) a synthetically generated population network based on the scale-free (S-F) model \cite{goh2001universal}, and (iii) two different real-world networks obtained from the Stanford Network Analysis Platform (SNAP) \cite{snapnets}.  We also adopt the weighted interaction rule as our model of how neighbors' beliefs influence respondents (see Section \ref{sec:problemformulation}).  We generate the weights randomly; for each edge in the graph, we draw a random variable from the uniform distribution over $[0,1]$, and we subsequently normalize the entries so that the sum of the weights (i.e., the tendencies for a respondent to listen to each of his/her neighbor) equals one.

We employ the following three sampling strategies. For the independent set sampling strategy, we generate an independent set of size 10\% of the population using Algorithm \ref{alg: indpset}. For the cluster sampling strategy, we first detect the clusters using the Leiden algorithm \cite{traag2019louvain}, one that has been used in recent studies \cite{austin2024uncovering,kim2021community}. We randomly select complete or partial clusters, identified by the Leiden algorithm, where the total cluster sample size is 10\% of the population. For the random sampling strategy, we randomly select 10\% of the population. We note that the choice of sampling proportion of 10\%, as mentioned in \cite{united2008designing}, is purely illustrative, and policymakers can adjust the size based on their domain knowledge, constraints, or any available contextual data. By setting the total population to 500, we conduct a series of numerical experiments with different parameters as listed in the following table. We generate an empirical distribution of the $W_1$ distance between the initial and updated belief distributions across multiple settings by running the numerical experiments 500 times. We show the empirical distributions as a boxplot in the respective figures, some of which are in Appendix \ref{appendix:statresults}.
\begin{table}[h]
    \centering
    \begin{tabular}{|c|c|c|c|}
    \hline
    Boxplots & \makecell{Initial Belief \\ Distribution} & \makecell{Interaction \\ Rule} \\
    \hline
    Figures \ref{fig:ergraphbeta22}, \ref{fig:fb22}, \ref{fig:SF22},  \ref{fig:email22} & Beta(2,2) & Average\\
    \hline
    Figures \ref{fig:ergraphbeta22weighted}, \ref{fig:fb22weighted}, \ref{fig:SF22weighted}, \ref{fig:email22weighted} & Beta(2,2) & Weighted\\
    \hline
    Figures \ref{fig:ergraphbeta25},  \ref{fig:fb25}, \ref{fig:SF25}, \ref{fig:email25} & Beta(2,5) & Average\\
    \hline
    Figures \ref{fig:ergraphbeta25weightedaverage}, \ref{fig:fb25weighted}, \ref{fig:SF25weighted}, \ref{fig:email25weighted} & Beta(2,5) & Weighted\\
    \hline
    Figures \ref{fig:ergraphstandardnormal}, \ref{fig:fbstandardnormal}, \ref{fig:SFstandardnormal}, \ref{fig:emailstandardnormal} & Normal(0,1) & Average\\
    \hline
    Figures \ref{fig:ergraphstandardnormalweighted}, \ref{fig:fbstandardnormalweighted}, \ref{fig:SFstandardnormalweighted}, \ref{fig:emailstandardnormalweighted} & Normal(0,1) & Weighted\\
    \hline
    \end{tabular}
    \caption{Different parameters on initial belief distribution and interaction mechanism over various graphs}
    \label{tab: diffparameters}
    \vspace{-15pt}
\end{table}

\textbf{E-R Model}.
In the first set of experiments,  we consider the E-R model.  We conduct the numerical experiments across different edge probabilities, specifically, $\frac{1}{|V|^2}, \frac{1}{|V|^{1.5}}, \frac{1}{|V|}, \frac{1}{|V|^{0.5}}, \frac{1}{|V|^{0.25}}, $ $\frac{1}{|V|^{0.125}}$, where $|V| = 500$.  As a reference, edges start to emerge at the threshold of $\frac{1}{|V|^2}$, while a giant component emerges at the threshold of $\frac{1}{|V|}$, in the limit \cite{jackson2008social}.
We note that as the E-R network becomes denser, the size of each cluster becomes larger and the size of the independent set sample inevitably decreases, as observed in Theorem \ref{thm: upperboundindpset}. We provide the number of selected clusters and the average size of the independent sets across varying densities of the population network in Table \ref{tab:descriptive} in \ref{appendix:statresults}. We provide the boxplots of the $W_1$ distance for each of the six probabilities in the following figures, as well as several statistical results such as the mean $W_1$ distance and the Kolmogorov-Smirnov (K-S) test statistics in Tables \ref{tab:avdister1/n2} to \ref{tab:K-Sstandardnormal}, in Appendix \ref{appendix:statresults}, respectively.
\begin{figure}[H]
    \centering
    \includegraphics[width=\linewidth]{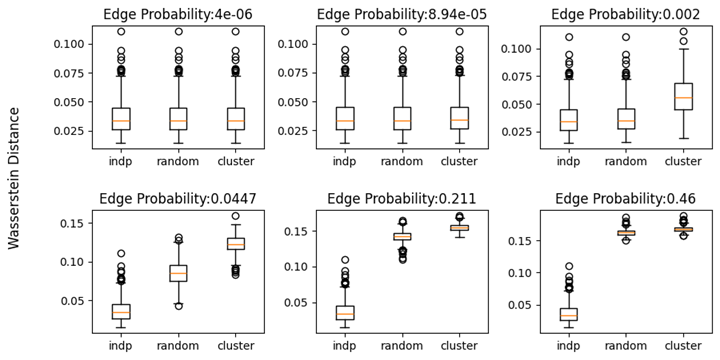}
    \caption{Initial distribution: Beta(2,2), Average interaction rule under different edge probabilities of the E-R model}\label{fig:ergraphbeta22}
    \vspace{-15pt}
\end{figure}
\begin{figure}[H]
    \centering
    \includegraphics[width=\linewidth]{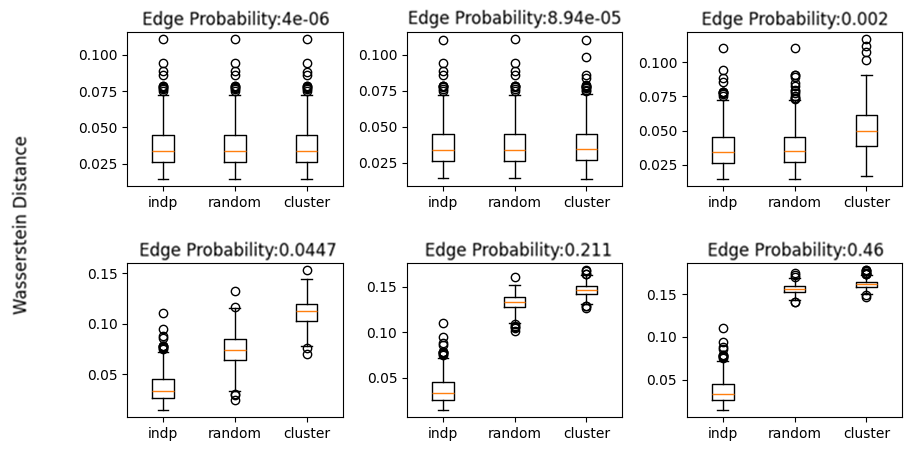}
    \caption{Initial distribution: Beta(2,2), Weighted interaction rule under different edge probabilities of the E-R model}\label{fig:ergraphbeta22weighted}
    \vspace{-18pt}
\end{figure}
\begin{figure}[H]
    \centering
    \includegraphics[width=\linewidth]{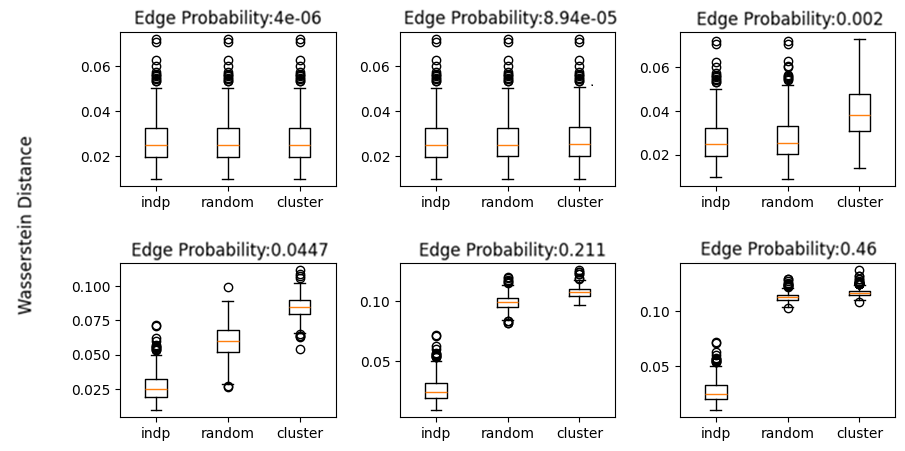}
    \caption{Initial distribution: Beta(2,5), Average interaction rule under different edge probabilities of the E-R model}\label{fig:ergraphbeta25}
    \vspace{-15pt}
\end{figure}
\begin{figure}[H]
    \centering
    \includegraphics[width=\linewidth]{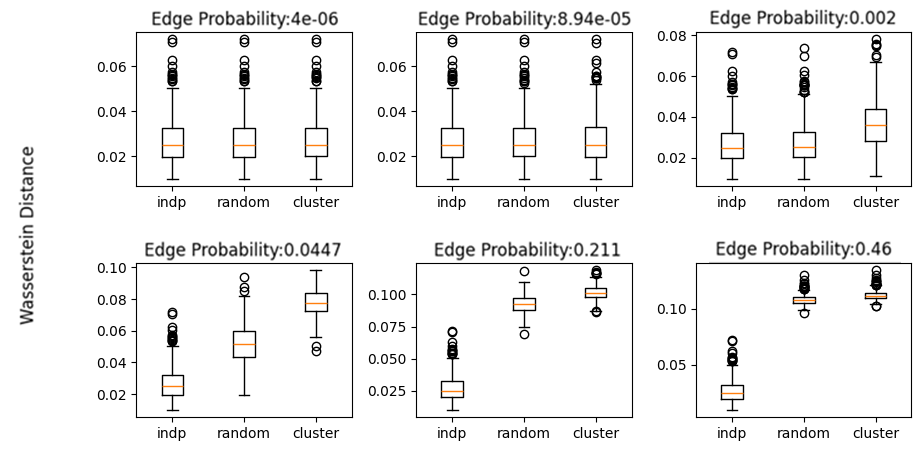}
    \caption{Initial distribution: Beta(2,5), Weighted interaction rule under different edge probabilities of the E-R model}\label{fig:ergraphbeta25weightedaverage}
    \vspace{-18pt}
\end{figure}
\begin{figure}[H]
    \centering
    \includegraphics[width=\linewidth]{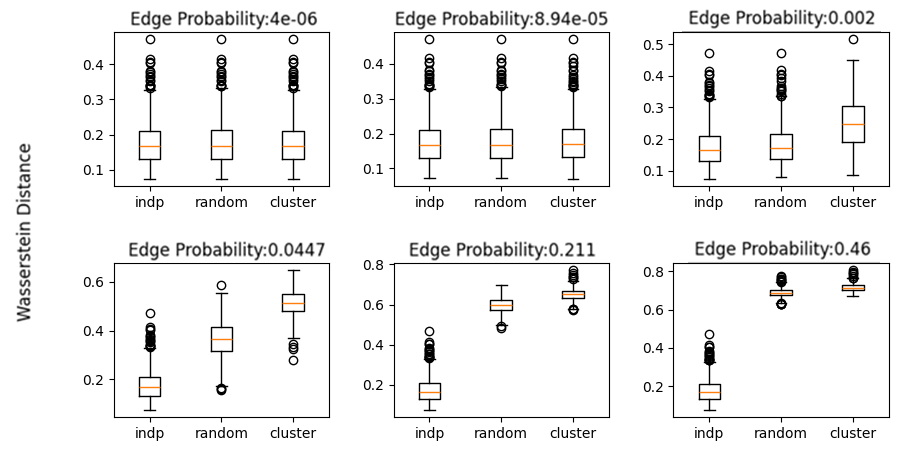}
    \caption{Initial distribution: N(0,1), Average interaction rule under different edge probabilities of the E-R model}\label{fig:ergraphstandardnormal}
    \vspace{-15pt}
\end{figure}
\begin{figure}[H]
    \centering
    \includegraphics[width=\linewidth]{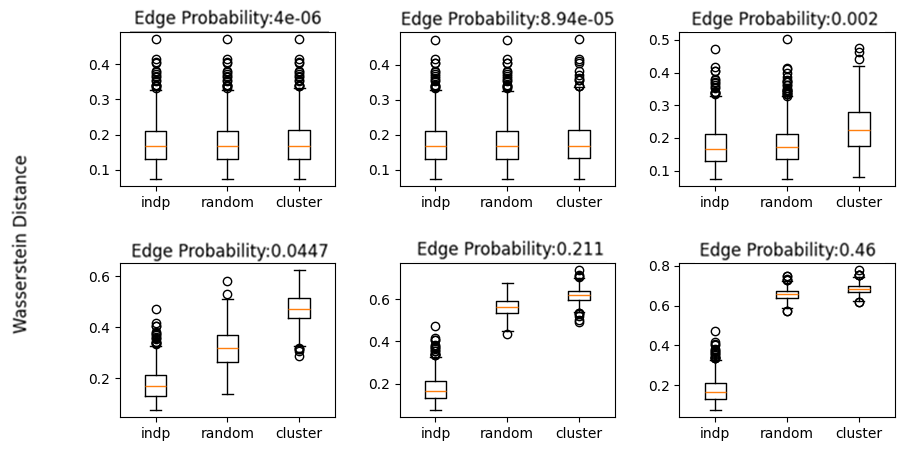}
    \caption{Initial distribution: N(0,1), Weighted interaction rule under different edge probabilities of the E-R model}\label{fig:ergraphstandardnormalweighted}
    \vspace{-12pt}
\end{figure}

Our first observation is that the sample size effect dominates when the graph is sparse, whereas the interaction effect dominates when the graph is dense.  The main implication of our observation is, if the population network is extremely sparse and there is minimal interaction effect, a random sampling strategy is preferred, as policymakers need a substantial independent set size to achieve a similar estimation performance, which can be difficult to obtain due to resources and feasibility constraints. This insight follows from the findings that when the edge probabilities are small, all three sampling strategies have a small mean $W_1$ distance. When the network is sparse, the resulting subgraphs after sampling have few edges. Then, the random and cluster samples largely consist of isolated vertices, and there is minimal interaction effect. Hence, the random and cluster samples perform almost as well as an independent set sample, and all three sampling strategies accurately estimate the initial belief distribution with small error. As the network becomes denser, the size of the independent set sample decreases. Thus, this affects the estimation of the initial belief distribution, as reflected in the increase of the mean $W_1$ distance. The random and cluster sampling strategies experience a larger increase in the mean $W_1$ distance as there are more interactions among the sampled respondents, arising from a denser subgraph. 

Our second observation is that the network density may reduce variability in the $W_1$ distance. 
Specifically, we note that the variance of the empirical $W_1$ distance distribution of the random and cluster samples decreases as the edge probabilities increase. We postulate that, as respondents have more neighbors, most tend to reach a consensus after interacting and share similar beliefs across the runs, thus contributing to similar empirical $W_1$ distances.

A closely related point is that the interaction mechanism under a dense network may affect the shift in beliefs.
Specifically, we observe the slightly larger mean $W_1$ distance for the random and cluster sampling strategies under the average interaction rule, as compared to those of the weighted counterparts (see Tables \ref{tab:avdister1/n2} to \ref{tab:avdister1/n0.125} in Appendix \ref{appendix:statresults}). These differences in the empirical $W_1$ distance distributions are also significant at the 1\% level of significance (see Tables \ref{tab:K-Sbeta2,2} to \ref{tab:K-Sstandardnormal} in Appendix \ref{appendix:statresults}). As the network density increases, respondents tend to have more neighbors and obtain similar beliefs after interacting under the average interaction rule. However, the information about the differences between their initial and updated opinions is lost when the respondents reach a consensus, thus resulting in a larger deviation in beliefs and $W_1$ distance. Unlike the average interaction rule, respondents have disproportionate influence on each other under the weighted interaction rule, and they are less likely to reach a consensus. This could then explain the smaller mean $W_1$ distance observed.

Our third observation is that the variance in the initial belief distribution (i.e., the diversity in opinions) may have an impact on how much opinion changes.  Specifically, we observe a smaller mean $W_1$ distance for the Beta(2, 5) distribution, as compared to the Beta(2, 2) distribution, across all threshold probabilities (see Tables \ref{tab:avdister1/n2} to \ref{tab:avdister1/n0.125} in Appendix \ref{appendix:statresults}).  The Beta(2, 5) distribution has a smaller variance than the Beta(2, 2) distribution. Hence, respondents’ beliefs are less likely to undergo substantial belief shifts following the interactions within themselves. Thus, there would be a smaller difference between the updated and the initial belief distributions and a smaller mean $W_1$ distance. 


\textbf{S-F Model}.
In the second set of experiments, we consider the S-F model where the generated networks have a scale-free property, i.e., the graph has a power law degree distribution. We consider a scale factor of 2.5, as many scale-free networks have scaling exponents between 2 and 3 \cite{goh2001universal,goh2002classification,chen2004modeling,barabási2016network}. We set the network size for the S-F models to 998, which is close to the expected size of the E-R graph with edge probability $\frac{1}{|V|}$, where $|V| = 500$. We note that the average number of clusters selected is about 4.8 across each setting. 
We provide several statistical results, such as the mean $W_1$ distance, the K-S test statistics and the boxplots of the empirical $W_1$ distance distributions in Tables \ref{tab:SFgraph}, \ref{tab:K-SSF} and Figure \ref{fig:overallSF}, in Appendix \ref{appendix:statresults}, respectively.

\begin{table}[H]
    \centering
    \resizebox{\columnwidth}{!}{%
    \begin{tabular}{|c|c|c|c|c|c|}
    \hline
    \multirow{2}{*}{Distribution} & \multirow{2}{*}{\makecell{Interaction \\ Rule}} & \multicolumn{2}{c|}{Random} & \multicolumn{2}{c|}{Cluster} \\
    \cline{3-6}
     & & E-R & S-F & E-R & S-F\\
    \hline
    \multirow{2}{*}{Beta(2,2)} & Average & 0.038 & 0.043 & 0.057 & 0.083\\
    \cline{2-6}
     & Weighted & 0.038 & 0.041 & 0.051 & 0.072\\
    \hline
    \multirow{2}{*}{Beta(2,5)} & Average  & 0.028 & 0.031 & 0.039 & 0.058\\
    \cline{2-6}
     & Weighted & 0.027 & 0.030 & 0.036 & 0.051\\
    \hline
    \multirow{2}{*}{Normal(0,1)} & Average  & 0.185 & 0.196 & 0.251 & 0.307\\
    \cline{2-6}
     & Weighted & 0.183 & 0.195 & 0.229 & 0.311\\
    \hline
    \end{tabular}}
    \caption{Comparing empirical mean $W_1$ distance between E-R and S-F networks}
    \label{tab:compareersf}
\end{table}

From Table \ref{tab:compareersf}, we note a larger mean $W_1$ distance for the random and cluster sampling strategies in the S-F networks as compared to the E-R networks. This is likely due to the network topology, as the degree distribution of E-R networks follows a binomial distribution, whereas S-F networks exhibit a power law degree distribution, where there are few nodes with high degree and a large number of nodes with low degree. Respondents with fewer neighbors experience a proportionately stronger influence from each individual neighbor, making their beliefs more susceptible to change. This could result in a large shift in their beliefs after interacting with neighbors, especially those who hold very different opinions, thereby resulting in a larger mean $W_1$ distance. This illustrates that the degree heterogeneity in the sample contributes to the extent of deviation in beliefs. As many real-world networks exhibit scale-free characteristics, many respondents tend to have a few neighbors, and these respondents can be sensitive to extreme opinions. Hence, policymakers need to recognize that such network topology can contribute to a large deviation in policy beliefs, which can affect the estimation of policy effectiveness.

\looseness=-1
\textbf{Undirected Real-World Network}.  In the third set of experiments, we use the \code{ego-Facebook} dataset from SNAP \cite{snapnets}, which consists of friends lists from Facebook.  The purpose of this experiment is to investigate whether our findings also apply to real-world networks.  The network contains 4039 vertices and 88234 edges, and is an undirected network (friends are mutual).  We observe that this network has \emph{small-world} properties as it has a moderately high average clustering coefficient of 0.6055 and a small average path length of 3.69. Obtaining a sizeable independent set sample is difficult, as this requires some knowledge about existing relationships among the population. Hence, we mimic the situation where policymakers do not have much resources and can only access a small independent set sample, say of size 30. We note that the average number of clusters selected is about 2.2 across each setting. We provide several statistical results, such as the mean $W_1$ distance and the K-S test statistics, in Tables \ref{tab:fbundirected} and \ref{tab:K-Sfacebook} in Appendix \ref{appendix:statresults}, respectively.

\begin{figure}[h]
     \centering
     \begin{subfigure}[b]{0.22\textwidth}
         \centering
         \includegraphics[width=\textwidth]{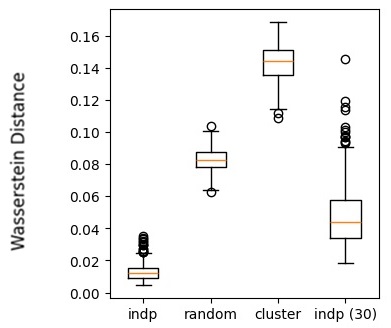}
         \caption{ }
         \label{fig:fb22}
     \end{subfigure}
     \hspace{-0.1cm}
     \begin{subfigure}[b]{0.217\textwidth}
         \centering
         \includegraphics[width=\textwidth]{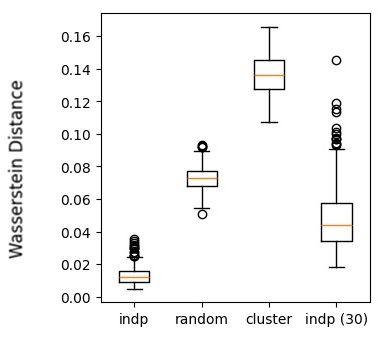}
         \caption{ }
         \label{fig:fb22weighted}
     \end{subfigure}
     \begin{subfigure}[b]{0.21\textwidth}
         \centering
         \includegraphics[width=\textwidth]{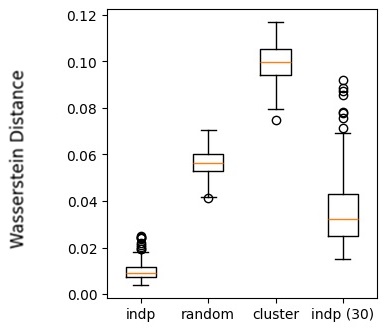}
         \caption{ }
         \label{fig:fb25}
     \end{subfigure}
     \hspace{0.1cm}
     \begin{subfigure}[b]{0.21\textwidth}
         \centering
         \includegraphics[width=\textwidth]{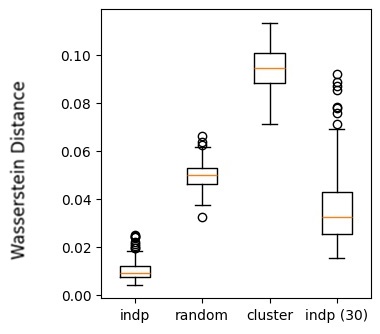}
         \caption{ }
         \label{fig:fb25weighted}
     \end{subfigure}
     \begin{subfigure}[b]{0.216\textwidth}
         \centering
         \includegraphics[width=\textwidth]{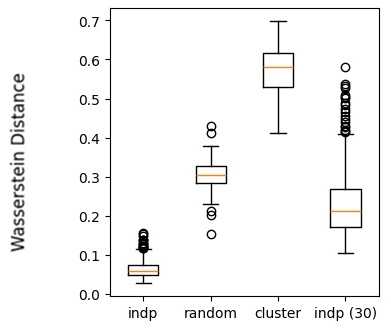}
         \caption{ }\label{fig:fbstandardnormal}
     \end{subfigure}
     \hspace{0.1cm}
     \begin{subfigure}[b]{0.21\textwidth}
         \centering
         \includegraphics[width=\textwidth]{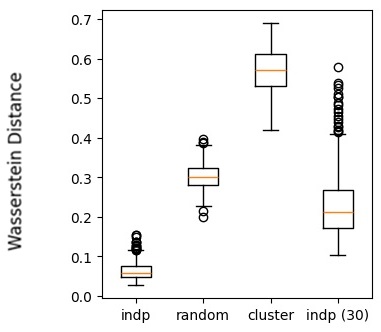}
         \caption{} \label{fig:fbstandardnormalweighted}
     \end{subfigure}
     \caption{Using \code{ego-Facebook} dataset from SNAP}        \label{fig:overallfb}
\end{figure}

From Figure \ref{fig:overallfb}, we note that a similar trend in the empirical $W_1$ distance distributions across the three sampling strategies, as compared to our previous findings from generated networks.
This validates the tradeoff between sample size and interaction effects inherent to the sampling strategies. Furthermore, from Figure \ref{fig:overallfb} and Table \ref{tab:fbundirected} in Appendix \ref{appendix:statresults}, the mean $W_1$ distance observed from the random sample is larger than that of the small independent set sample, though it has a smaller variance due to the larger sample size. This underscores a practical tradeoff in allocating resources to improve the estimation of population beliefs, where resources can be directed toward either obtaining a larger independent set sample to reduce the sample size effect or adopting alternative measures, such as enhancing policy communication, so that more people share similar beliefs, thereby reducing the interaction effect.

In addition, the moderately high clustering coefficient (0.6055) coupled with low density (about 0.01) suggests the presence of local clusters, which facilitates a significant level of interaction among respondents and the formation of consensus within their local communities. However, each of these communities could arrive at different consensuses, thereby contributing to the larger mean $W_1$ distance observed in the random and cluster samples (see Table \ref{tab:fbundirected} in Appendix \ref{appendix:statresults}).

\textbf{Directed Real-World Network}.  In the fourth set of experiments, we use the directed \code{email-Eu-core} network dataset from SNAP \cite{snapnets}.  The purpose of this experiment is to understand how directionality affects our conclusions in the context of a real network.  The network dataset is generated using email data from a large European research institution, with 81306 vertices and 1768149 edges in this directed network. Some descriptive statistics include an average clustering coefficient of 0.3994 and an average path length of 2.544 in the giant component. Similar to the previous numerical experiment, we include the case where we have a small independent set sample of size 30. We note that the average number of clusters selected is about 1.28 across each setting. We provide several statistical results, such as the mean $W_1$ distance, the K-S test statistics and the boxplots of the empirical $W_1$ distance distributions in Tables \ref{tab:emaildirected}, \ref{tab:K-Sdirectedemail} and Figure \ref{fig:overallemail}, in Appendix \ref{appendix:statresults}, respectively.

To investigate how directionality affects the estimation performance, we compare an undirected version of the \code{email-Eu-core} network dataset against the original directed version. We generate the undirected \code{email-Eu-core} network by removing the directionality in the edges. We state the mean $W_1$ distance for the undirected network and the K-S statistics between the undirected and directed networks across different settings in Tables \ref{tab:emailundirected} and \ref{tab:kstestemail}, in Appendix \ref{appendix:statresults}, respectively.

\begin{table}[h]
    \centering
    \resizebox{\columnwidth}{!}{%
    \begin{tabular}{|c|c|c|c|c|c|}
    \hline
    \multirow{2}{*}{Distribution} & \multirow{2}{*}{\makecell{Interaction \\ Rule}} & \multicolumn{2}{c|}{Random} & \multicolumn{2}{c|}{Cluster} \\
    \cline{3-6}
     & & Directed & Undirected & Directed & Undirected\\
    \hline
    \multirow{2}{*}{Beta(2,2)} & Average & 0.084 & 0.069 & 0.126 & 0.122\\
    \cline{2-6}
     & Weighted & 0.068 & 0.065 & 0.117 & 0.116\\
    \hline
    \multirow{2}{*}{Beta(2,5)} & Average  & 0.058 & 0.047 & 0.087 & 0.084\\
    \cline{2-6}
     & Weighted & 0.047 & 0.044 & 0.081 & 0.080\\
    \hline
    \multirow{2}{*}{Normal(0,1)} & Average  & 0.287 & 0.276 & 0.496 & 0.485\\
    \cline{2-6}
     & Weighted & 0.286 & 0.275 & 0.495 & 0.481\\
    \hline
    \end{tabular}}
    \caption{Comparing mean $W_1$ distance between directed and undirected  \code{email-Eu-core} network}
    \label{tab:compareeucore}
\end{table}

From Table \ref{tab:compareeucore}, we observe that the mean $W_1$ distances for the directed version are generally smaller than the undirected counterparts, where the differences are significant at the 5\% significance level (see Table \ref{tab:emailundirected} in Appendix \ref{appendix:statresults}). Since respondents only listen to his or her influencers in their directed neighborhood, rather than engaging in bilateral interaction in undirected networks, the directed neighborhood is generally smaller than its undirected counterpart. Hence, each neighbor in the directed neighborhood has a larger influence on a respondent. This could lead to a larger shift in belief, especially when these neighbors have largely different opinions, thus resulting in a larger $W_1$ distance.
\section{Conclusion and Future Directions}


In our work, we examine how the population's opinion towards a policy changes as people interact among themselves and how these deviations influence the estimation of the initial population beliefs. We consider various factors that impact the changes in beliefs, such as sampling strategies, the underlying network topology and the nature of interaction among respondents. For instance, the (un)willingness of respondents to listen to their neighbors affects the formation of consensus, which then affects the deviation from their original belief. By quantifying these shifts analytically and numerically, we illustrate a tradeoff between two important factors: the sampling size and interaction effects. 


\looseness=-1
Through the design of interaction mechanisms and sampling strategies that model real-world policy evaluation studies, our work aims to provide a conceptual framework for policymakers to estimate the initial policy responses and quantify the deviation from the observed data, which has clear practical implications in evaluating policy effectiveness. Furthermore, as highlighted in our work, different sampling strategies introduce varying degrees of interaction between respondents. To obtain accurate estimates of the population's beliefs, policymakers not only need to adopt an appropriate sampling strategy -- though its choice is often shaped by practical constraints such as cost and feasibility -- but also consider complementary measures to manage the potential interaction effect. For example, adopting clearer and more effective policy communication could help respondents form similar beliefs, thereby reducing the deviation due to the interaction.

In the following, we discuss some assumptions and limitations of our model as well as potential future directions.  First, our model does not consider the natural tendency of respondents to change their own beliefs and incorporate external factors that can potentially shape their perspectives. Although a controlled and closed setting may be a strict assumption under various contexts, this stylized environment allows policymakers to recognize and quantify the tradeoff between the sampling size and interaction effects. 
Second, our model assumes that respondents interact over a given contact network and update their beliefs according to specified interaction parameters. As respondents assign different weights to their neighbors based on their views and characteristics, these parameters may be misspecified and further robustness analyses may be needed to calibrate them. Moreover, the model does not explicitly account for the homophilous nature of real-world social networks, such as assortativity in beliefs and characteristics. As respondents are more likely to interact with neighbors who share similar beliefs, policymakers can consider alternative models, such as the bounded confidence model, to better reflect the nature of these interactions. However, incorporating such dynamics would increase the complexity of the model, though it could reduce the magnitude of interaction effects, thus yielding tighter bounds and more accurate estimates. 
Third, our work only considers static networks to simplify the estimation of the initial belief distribution. We acknowledge that beliefs can change over time due to various reasons. Hence, another avenue for future research is to incorporate dynamic network processes such as changing levels of interaction over time or homophily, where respondents form new edges due to similar policy responses. This, in turn, influences the spread of information and attitudes within social networks. Depending on the network topology and the characteristics of the respondents, it could result in either consensus or disagreement over time \cite{acemoglu2011opinion}. With a possibly larger interaction effect, this would likely increase the empirical mean $W_1$ distance. 
Fourth, we note the difficulty of estimating the exact shape and moments of the initial belief distribution under limited data and unspecified network dynamics. An interesting future direction is to provide such structural characteristics of the belief distribution to understand the nuances in the beliefs and behaviors of the population better. Furthermore, policymakers can investigate the dynamics within the respondents' neighborhood as the network topology influences how beliefs evolve. 
Lastly, we acknowledge that, in reality, policy beliefs can spread across multiple networks, such as physical contact networks and social media platforms, with differing interaction patterns. As future work, policymakers can incorporate various offline and online networks to capture more realistic interaction dynamics and achieve more reliable estimates.

\section*{Acknowledgments}
The authors would like to thank Dr. Prasanta Bhattacharya for his comments on the manuscript.





\newpage

\appendices

\section{Proof of Claim \ref{claim:assumptionfewdegree}}\label{appendix: proofofclaim}

\begin{proof}
{\em Bounding $|N_v|$:} For a given vertex $v$, let $\mathbf{D}_v$ be the degree of $v$. Under an Erd\"{o}s-Renyi random graph model, $\mathcal{R}(|V|,p)$, $\mathbf{D}_v$ is the sum of $|V|-1$ independent Bernoulli random variables with success probability $p$. Since the degree of any vertex under $\mathcal{R}(|V|,p)$ follows the same binomial distribution, the average degree $\langle d \rangle$ of the graph is $(|V|-1)p$, which is equivalent to $\E{\mathbf{D}_v}$. Hence, the tail probability of $\mathbf{D}_v$ exceeding twice the average degree, is given as
\vspace{-5pt}
\begin{equation*}
\begin{aligned}
\textstyle \Prob{\mathbf{D}_v \geq 2\E{\mathbf{D}_v}} &= \Prob{\mathbf{D}_v - \E{\mathbf{D}_v} \geq \E{\mathbf{D}_v}}\\
& \textstyle \overset{(a)}{\leq} \exp ( {-\frac{2\E{\mathbf{D}_v}^2}{|V|-1}} )\\
&=\exp (-2p^2(|V|-1))
\end{aligned}
\end{equation*}
where (a) is by the Hoeffding's inequality. \\

\vspace{-10pt}



\vspace{-5pt}
\looseness-1
{\em Bounding $|M_v|$:} To bound the number of vertices that are exactly 2-hops away from $v$, we consider the number of 2-stars with $v$ as a leaf vertex. We denote $\mathbf{D}_{2v}$ as the random variable for the number of 2-stars with $v$ as a leaf vertex. Then, the expected number of 2-stars is given as
\vspace{-5pt}
$$\E{\mathbf{D}_{2v}} = (|V|-1)(|V|-2)p^2 \approx \langle d \rangle^2$$
The random variable $\mathbf{D}_{2v}$ is represented as $\sum_{j}\mathds{1}_{\tau_j}$, where $\tau_j$ is the event that the $j$-th triple of vertices $(v_{j1},v_{j2},v_{j3})$, form a 2-star with $v$ as a leaf vertex. Then, $\E{(\mathbf{D}_{2v})^2}$ is given as,
\vspace{-8pt}
\begin{equation*}
\begin{aligned}
\textstyle \E{(\mathbf{D}_{2v})^2} &= \textstyle \E{(\sum_j\mathds{1}_{\tau_j})^2}\\ &= \textstyle \E{\sum_j\mathds{1}_{\tau_j}} + \sum_{j \neq k}\E{\mathds{1}_{\tau_j}\mathds{1}_{\tau_k}}
\end{aligned}
\end{equation*}
As $\E{\mathds{1}_{\tau_j}\mathds{1}_{\tau_k}} = \Prob{\mathds{1}_{\tau_j}\cap \mathds{1}_{\tau_k}}$ is the probability that both triples $j$ and $k$ have a 2-star with $v$ as a leaf vertex. If $\tau_j \cap \tau_k = v$, then $\Prob{\tau_j \cap \tau_k} = p^4$; if $\tau_j$ and $\tau_k$ have 2 common vertices, including $v$ and no common edge, then $\Prob{\tau_j \cap \tau_k} = p^4$ and if $\tau_j$ and $\tau_k$ have 2 common vertices, including $v$ and 1 common edge, then $\Prob{\tau_j \cap \tau_k} = p^3$. The total number of the pairs of triples $j$ and $k$ with $v$ as a common vertex is $4\binom{|V|-1}{2}\binom{|V|-3}{2}$, with 2 common vertices including $v$ and a common edge is $2\binom{|V|-1}{2}(|V|-3)$ and with two common vertices including $v$ and no common edge is $4\binom{|V|-1}{2}(|V|-3)$. Summing, we have
\begin{equation*}
\resizebox{\hsize}{!}{$
\begin{aligned}
\sum_{j \neq k}\E{\mathds{1}_{\tau_j}\mathds{1}_{\tau_k}} & = \sum_{j \neq k}\Prob{\mathds{1}_{\tau_j}\cap \mathds{1}_{\tau_k}}\\ 
&= 4\tbinom{|V|-1}{2}\tbinom{|V|-3}{2}p^4 + 2\tbinom{|V|-1}{2}(|V|-3)p^3 \\
& \quad + 4\tbinom{|V|-1}{2}(|V|-3)p^4
\end{aligned}
$}
\end{equation*}
Using the facts that $\binom{|V|-1}{2} \sim \binom{|V|-3}{2} \sim \frac{|V|^2}{2}$ and $|V|-3 \sim |V|$, we find that $\sum_{j \neq k}\E{\mathds{1}_{\tau_j}\mathds{1}_{\tau_k}} = (1+o(1))\E{\mathbf{D}_{2v}}^2$, for any given vertex $v$. Then, the tail probability of $\mathbf{D}_{2v}$ exceeding twice its mean is given as
\vspace{-5pt}
\begin{equation*}
\begin{aligned}
\Prob{\mathbf{D}_{2v} \geq 2\E{\mathbf{D}_{2v}}} & \textstyle \overset{(a)}{\leq} \frac{Var(\mathbf{D}_{2v})}{\E{\mathbf{D}_{2v}}^2} \\
& \textstyle = \frac{\E{(\mathbf{D}_{2v})^2} - \E{\mathbf{D}_{2v}}^2}{\E{\mathbf{D}_{2v}}^2} \\
& \textstyle = \frac{\E{\mathbf{D}_{2v}} + \sum_{j \neq k}\E{\mathds{1}_{\tau_j}\mathds{1}_{\tau_k}} - \E{\mathbf{D}_{2v}}^2}{\E{\mathbf{D}_{2v}}^2}\\
& \textstyle = \frac{1}{\E{\mathbf{D}_{2v}}} + o(1)\\
& \textstyle = \frac{1}{(|V|-1)(|V|-2)p^2} + o(1)
\end{aligned}
\end{equation*}
where (a) is by the Chebyshev inequality. Combining both probabilities with the union bound and using the fact that $|V|-2 \sim |V|-1 \sim |V|$, we have 
\begin{equation*}
\resizebox{.9\hsize}{!}{$
\begin{aligned}
&\textstyle\Prob{\bigcap_{v \in V} (\mathbf{L}_v \leq 2(\langle d \rangle + \langle d \rangle^2))}\\
& \textstyle \overset{(a)}{\geq} 1- \sum_{v \in V} \Prob{\mathbf{D}_v + \mathbf{D}_{2v} \geq 2(\langle d \rangle + \langle d \rangle^2)}\\
& \textstyle \overset{(b)}{\geq} 1 - \sum_{v \in V}\Prob{\mathbf{D}_v \geq 2\langle d \rangle} - \sum_{v \in V}\Prob{\mathbf{D}_{2v} \geq 2\langle d \rangle^2}\\
& \textstyle = 1- |V|\exp (-2p^2(|V|-1)) - \frac{|V|}{(|V|-1)(|V|-2)p^2} - o(1)\\
& \textstyle \gtrsim 1 - \frac{2}{|V|p^2}
\end{aligned}
$}
\end{equation*}
where (a) and (b) are by the union bound. Then, Assumption \ref{assume:fewdegree} holds with high probability for sufficiently large $|V|$.
\end{proof}
\section{Algorithm for Sampling Independent Sets}\label{appendix: alg}
In this section, we present a pseudocode that we employ to obtain our independent set samples. We adopt an iterative approach to obtain the independent set for our numerical studies, as used in other studies \cite{blelloch2012greedy}. We provide the following definition to complement our algorithm. 

\begin{definition}
A graph $G'(V',E')$ is a \emph{subgraph} of $G(V,E)$ where $V' \subseteq V$ and $E' \subseteq E$. For a subset of nodes $V' \subseteq V$, the subgraph of $G$ induced by $V'$ is the subgraph $G'(V',E')$, where $E' = E \cap \binom{V'}{2}$. A \emph{vertex-induced subgraph} contains a subset of vertices, coupled with edges whose endpoints are both in the subset.
\end{definition}

We also add a size limit parameter to the pseudocode to provide flexibility in controlling the size of the independent set. For simulations where we do not need to constrain the size of the independent set sample, we can set $S = |V(G)|$.
\begin{algorithm}[H]
\caption{Iterative Independent Set Selection}\label{alg: indpset}
\textbf{Input: } Graph $G$, Size $S$ \\
\textbf{Output: } Independent Set $I$
\begin{algorithmic}[1]
\While{$G \neq \emptyset$}
    \State Randomly select a vertex $v \in V(G)$
    \State $I \leftarrow I \cup v$
    \State $G \leftarrow \Tilde{G}$, where $\Tilde{G}$ is the vertex-induced subgraph of $V(G) \backslash (v \cup N_v)$
    \If{$|I| > S$}
        \State break
    \EndIf
\EndWhile
\State \Return $I$
\end{algorithmic}
\end{algorithm}
\section{Expansion of Corollary \ref{cor:randomsamplinggeneral}}\label{appendix: expansionofcor}
We recall that there are two independent elements of randomness in the problem context, the sampling strategy and the initial beliefs of the respondents. We first fix the sampled respondents, thus determining the resulting subgraph $\Tilde{G}$. Then, by the law of iterated expectation, we have
$$\E{W_1(\hat{F}_n,F)} = \E{\E{W_1(\hat{F}_n,F)|\Tilde{G}}}$$
From Proposition \ref{prop:randomsampling}, we have the upper bound of $\E{W_1(\hat{F}_n,F)|\Tilde{G}}$. As the degree of the selected respondents changes across various resulting subgraphs, we need to take the expectation over all possible selections of $n$ respondents. By Fubini's theorem and Jensen's inequality, the first and third rows in the upper bound of Proposition \ref{prop:randomsampling} are given as,
\[
\resizebox{\hsize}{!}{$
\begin{aligned}
& \textstyle \E{\frac{1}{n}\int_{\R}\sqrt{\sum_{i=1}^nF^{d_i+1}(t)(1-F^{d_i+1}(t)) + \frac{\sum_{r \neq s}^nM_{r,s}}{4}} dt}\\ 
\leq & \textstyle \frac{1}{n}\int_{\R}\sqrt{\E{\sum_{i=1}^nF^{d_i+1}(t)(1-F^{d_i+1}(t))} + \frac{\E{\sum_{r \neq s}^nM_{r,s}}}{4}} dt
\end{aligned}
$}
\]

For a fixed $t$, we note that $F^{d_i+1}(t)(1-F^{d_i+1}(t))$ is a random function as it depends on the degree of the sampled respondents. By taking expectation, we have the following expression,
\begin{equation*}
\resizebox{\hsize}{!}{$
\begin{aligned}
&\textstyle \E{\sum_{i=1}^nF^{d_i+1}(t)(1-F^{d_i+1}(t))} \\
& = 
\textstyle \frac{1}{\binom{|P|}{n}}\sum_{i=1}^{|P|}\sum_{\overline{d}_i = 0}^{d_i}\binom{d_i}{\overline{d}_i}\binom{|P|-d_i-1}{n-\overline{d}_i-1}F^{\overline{d}_i+1}(t)(1-F^{\overline{d}_i+1}(t))
\end{aligned}
$}
\end{equation*}
where $d_i$ is the degree of respondent $i$ in the population.\footnote{We follow the convention that if $k < 0$ then $\binom{n}{k} = 0$} Taking reference to Assumption \ref{assume:fewdegree}, we assume that $\E{\sum_{r \neq s}^nM_{r,s}} = O(\langle d \rangle^2n)$, where $\langle d \rangle$ is the average degree of the initial population graph and remain in the same order after the sampling process. 
For the second row in the upper bound, we similarly have
\begin{equation*}
\resizebox{\hsize}{!}{$
\begin{aligned}
& \textstyle \E{\sigma\sqrt{2(1-\rho^{F,\Tilde{F}})} + \sigma + \frac{\sigma}{n}\sum_{i=1}^n \frac{\sqrt{2(1-\rho^{F^{d_i+1},\Tilde{F}^{d_i+1}})}-1}{\sqrt{d_i+1}}} \\
& \textstyle \quad =  \sigma\sqrt{2(1-\rho^{F,\Tilde{F}})} + \sigma \\ 
& \textstyle \quad + \frac{\sigma}{n\binom{|P|}{n}}\sum_{i=1}^{|P|}\sum_{\overline{d}_i = 0}^{d_i}\tbinom{d_i}{\overline{d}_i}\tbinom{|P|-d_i-1}{n-\overline{d}_i-1}\frac{\sqrt{2(1-\rho^{F^{\overline{d}_i+1},\Tilde{F}^{\overline{d}_i+1}})}-1}{\sqrt{\overline{d}_i+1}}
\end{aligned}
$}
\end{equation*}
We then have the following bound after combining both terms and splitting the square-root term.
\section{Statistical Results}\label{appendix:statresults}
\begin{table}[H]
    \centering
    \resizebox{\columnwidth}{!}{%
    \begin{tabular}{|c|c|c|c|c|}
    \hline
    Probability & Distribution & \makecell{Interaction \\ Rule} & \makecell{Average \\ Independent \\ Set Size} & \makecell{Average \\ Number of \\ Clusters} \\
    \hline
    \multirow{6}{*}{$4 \times 10^{-6}$} & \multirow{2}{*}{Beta(2,2)} & Average & 50  & 49.938 \\
    \cline{3-5}
    & & Weighted & 50  & 49.952 \\
    \cline{2-5}
     & \multirow{2}{*}{Beta(2,5)} & Average & 50  & 49.956 \\
    \cline{3-5}
    & & Weighted & 50  & 49.944 \\
    \cline{2-5}
     & \multirow{2}{*}{Normal(0,1)} & Average & 50  & 49.958 \\
    \cline{3-5}
    & & Weighted & 50  & 49.956 \\
    \hline
    \multirow{6}{*}{$8.9 \times 10^{-5}$} & \multirow{2}{*}{Beta(2,2)} & Average & 50  & 48.938 \\
    \cline{3-5}
    & & Weighted & 50  & 48.872 \\
    \cline{2-5}
     & \multirow{2}{*}{Beta(2,5)} & Average & 50  & 48.970 \\
    \cline{3-5}
    & & Weighted & 50  & 48.916 \\
    \cline{2-5}
     & \multirow{2}{*}{Normal(0,1)} & Average & 50 & 48.900 \\
    \cline{3-5}
    & & Weighted & 50  & 48.836 \\
    \hline
    \multirow{6}{*}{$0.002$} & \multirow{2}{*}{Beta(2,2)} & Average & 50  & 27.018 \\
    \cline{3-5}
    & & Weighted & 50  & 26.634 \\
    \cline{2-5}
     & \multirow{2}{*}{Beta(2,5)} & Average & 50  & 25.970 \\
    \cline{3-5}
    & & Weighted & 50  & 26.746 \\
    \cline{2-5}
     & \multirow{2}{*}{Normal(0,1)} & Average & 50  & 26.810 \\
    \cline{3-5}
    & & Weighted & 50  & 26.428 \\
    \hline
    \multirow{6}{*}{$0.045$} & \multirow{2}{*}{Beta(2,2)} & Average & 50  & 1.458 \\
    \cline{3-5}
    & & Weighted & 50  & 1.472 \\
    \cline{2-5}
     & \multirow{2}{*}{Beta(2,5)} & Average & 50  & 1.434 \\
    \cline{3-5}
    & & Weighted & 50  & 1.444 \\
    \cline{2-5}
     & \multirow{2}{*}{Normal(0,1)} & Average & 50  & 1.478 \\
    \cline{3-5}
    & & Weighted & 50  & 1.432 \\
    \hline
    \multirow{6}{*}{$0.211$} & \multirow{2}{*}{Beta(2,2)} & Average & 20.420 & 1.224 \\
    \cline{3-5}
    & & Weighted & 20.324  & 1.212 \\
    \cline{2-5}
     & \multirow{2}{*}{Beta(2,5)} & Average & 20.484  & 1.194 \\
    \cline{3-5}
    & & Weighted & 20.462  & 1.238 \\
    \cline{2-5}
     & \multirow{2}{*}{Normal(0,1)} & Average & 20.418  & 1.212 \\
    \cline{3-5}
    & & Weighted & 20.368  & 1.204 \\
    \hline
    \multirow{6}{*}{$0.460$} & \multirow{2}{*}{Beta(2,2)} & Average & 9.578  & 1.084 \\
    \cline{3-5}
    & & Weighted & 9.466  & 1.086 \\
    \cline{2-5}
     & \multirow{2}{*}{Beta(2,5)} & Average & 9.570  & 1.056 \\
    \cline{3-5}
    & & Weighted & 9.482  & 1.062 \\
    \cline{2-5}
     & \multirow{2}{*}{Normal(0,1)} & Average & 9.598  & 1.090 \\
    \cline{3-5}
    & & Weighted & 9.610  & 1.088 \\
    \hline
    \end{tabular}}
    \caption{Descriptives for E-R graphs}
    \label{tab:descriptive}
\end{table}

\begin{table}[H]
    \centering
    \resizebox{\columnwidth}{!}{%
    \begin{tabular}{|c|c|c|c|c|}
    \hline
    Distribution & Interaction Rule & Independent & Random & Cluster \\
    \hline
    \multirow{2}{*}{Beta(2,2)} & Average & 0.037 & 0.037 & 0.037 \\
    \cline{2-5}
     & Weighted & 0.037 & 0.037 & 0.037 \\
    \hline
    \multirow{2}{*}{Beta(2,5)} & Average & 0.027 & 0.027 & 0.027 \\
    \cline{2-5}
     & Weighted & 0.027 & 0.027 & 0.027 \\
    \hline
    \multirow{2}{*}{Normal(0,1)} & Average & 0.180 & 0.180 & 0.180 \\
    \cline{2-5}
     & Weighted & 0.180 & 0.180 & 0.180 \\
    \hline
    \end{tabular}}
    \caption{Mean $W_1$ distance for E-R graphs with edge probability $4 \times 10^{-6}$}
    \label{tab:avdister1/n2}
\end{table}

\begin{table}[H]
    \centering
    \resizebox{\columnwidth}{!}{%
    \begin{tabular}{|c|c|c|c|c|}
    \hline
    Distribution & Interaction Rule & Independent & Random & Cluster \\
    \hline
    \multirow{2}{*}{Beta(2,2)} & Average & 0.037 & 0.037 & 0.037 \\
    \cline{2-5}
     & Weighted & 0.037 & 0.037 & 0.037 \\
    \hline
    \multirow{2}{*}{Beta(2,5)} & Average & 0.027 & 0.027 & 0.027 \\
    \cline{2-5}
     & Weighted & 0.027 & 0.027 & 0.027 \\
    \hline
    \multirow{2}{*}{Normal(0,1)} & Average & 0.180 & 0.180 & 0.182 \\
    \cline{2-5}
     & Weighted & 0.180 & 0.180 & 0.181 \\
    \hline
    \end{tabular}}
    \caption{Mean $W_1$ distance for E-R graphs with edge probability $8.9 \times 10^{-5}$}
    \label{tab:avdister1/n1.5}
\end{table}

\begin{table}[H]
    \centering
    \resizebox{\columnwidth}{!}{%
    \begin{tabular}{|c|c|c|c|c|}
    \hline
    Distribution & Interaction Rule & Independent & Random & Cluster \\
    \hline
    \multirow{2}{*}{Beta(2,2)} & Average & 0.037 & 0.038 & 0.057 \\
    \cline{2-5}
     & Weighted & 0.037 & 0.038 & 0.051 \\
    \hline
    \multirow{2}{*}{Beta(2,5)} & Average & 0.027 & 0.028 & 0.039 \\
    \cline{2-5}
     & Weighted & 0.027 & 0.027 & 0.036 \\
    \hline
    \multirow{2}{*}{Normal(0,1)} & Average & 0.180 & 0.185 & 0.251 \\
    \cline{2-5}
     & Weighted & 0.180 & 0.183 & 0.229 \\
    \hline
    \end{tabular}}
    \caption{Mean $W_1$ distance for E-R graphs with edge probability 0.002}
    \label{tab:avdister1/n}
\end{table}

\begin{table}[H]
    \centering
    \resizebox{\columnwidth}{!}{%
    \begin{tabular}{|c|c|c|c|c|}
    \hline
    Distribution & Interaction Rule & Independent & Random & Cluster \\
    \hline
    \multirow{2}{*}{Beta(2,2)} & Average & 0.037 & 0.085 & 0.123 \\
    \cline{2-5}
     & Weighted & 0.037 & 0.075 & 0.111 \\
    \hline
    \multirow{2}{*}{Beta(2,5)} & Average & 0.027 & 0.059 & 0.084 \\
    \cline{2-5}
     & Weighted & 0.027 & 0.051 & 0.078 \\
    \hline
    \multirow{2}{*}{Normal(0,1)} & Average & 0.180 & 0.363 & 0.513 \\
    \cline{2-5}
     & Weighted & 0.180 & 0.316 & 0.472 \\
    \hline
    \end{tabular}}
    \caption{Mean $W_1$ distance for E-R graphs with edge probability 0.045}
    \label{tab:avdister1/n0.5}
\end{table}

\begin{table}[H]
    \centering
    \resizebox{\columnwidth}{!}{%
    \begin{tabular}{|c|c|c|c|c|}
    \hline
    Distribution & Interaction Rule & Independent & Random & Cluster \\
    \hline
    \multirow{2}{*}{Beta(2,2)} & Average & 0.037 & 0.143 & 0.154 \\
    \cline{2-5}
     & Weighted & 0.037 & 0.133 & 0.147 \\
    \hline
    \multirow{2}{*}{Beta(2,5)} & Average & 0.027 & 0.099 & 0.107 \\
    \cline{2-5}
     & Weighted & 0.027 & 0.092 & 0.102 \\
    \hline
    \multirow{2}{*}{Normal(0,1)} & Average & 0.180 & 0.600 & 0.651 \\
    \cline{2-5}
     & Weighted & 0.180 & 0.561 & 0.618 \\
    \hline
    \end{tabular}}
    \caption{Mean $W_1$ distance for E-R graphs with edge probability 0.211}
    \label{tab:avdister1/n0.25}
\end{table}

\begin{table}[H]
    \centering
    \resizebox{\columnwidth}{!}{%
    \begin{tabular}{|c|c|c|c|c|}
    \hline
    Distribution & Interaction Rule & Independent & Random & Cluster \\
    \hline
    \multirow{2}{*}{Beta(2,2)} & Average & 0.037 & 0.163 & 0.168 \\
    \cline{2-5}
     & Weighted & 0.037 & 0.156 & 0.161 \\
    \hline
    \multirow{2}{*}{Beta(2,5)} & Average & 0.027 & 0.113 & 0.117 \\
    \cline{2-5}
     & Weighted & 0.027 & 0.108 & 0.112 \\
    \hline
    \multirow{2}{*}{Normal(0,1)} & Average & 0.180 & 0.689 & 0.715 \\
    \cline{2-5}
     & Weighted & 0.180 & 0.656 & 0.684 \\
    \hline
    \end{tabular}}
    \caption{Mean $W_1$ distance for E-R graphs with edge probability 0.460}
    \label{tab:avdister1/n0.125}
\end{table}

\begin{table}[H]
    \centering
    \begin{tabular}{|c|c|c|}
    \hline
    Probability & Strategies & Statistic \\
    \hline
    \multirow{3}{*}{4 $\times 10^{-6}$} & Independent & 0.000 \\
    \cline{2-3}
    & Random & 0.002 \\
    \cline{2-3}
    & Cluster & 0.008 \\
    \hline
    \multirow{3}{*}{8.9 $\times 10^{-5}$} & Independent & 0.000 \\
    \cline{2-3}
    & Random & 0.010 \\
    \cline{2-3}
    & Cluster & 0.024 \\
    \hline
    \multirow{3}{*}{0.002} & Independent & 0.000 \\
    \cline{2-3}
    & Random & 0.036\\
    \cline{2-3}
    & Cluster & 0.138$^*$\\
    \hline
    \multirow{3}{*}{0.045} & Independent & 0.000\\
    \cline{2-3}
    & Random & 0.272$^*$\\
    \cline{2-3}
    & Cluster & 0.380$^*$\\
    \hline
    \multirow{3}{*}{0.211} & Independent & 0.000\\
    \cline{2-3}
    & Random & 0.464$^*$\\
    \cline{2-3}
    & Cluster & 0.492$^*$\\
    \hline
    \multirow{3}{*}{0.460} & Independent & 0.000\\
    \cline{2-3}
    & Random & 0.588$^*$\\
    \cline{2-3}
    & Cluster & 0.690$^*$\\
    \hline
    \end{tabular}
    \caption{Kolmogorov-D statistic from 2-sample K-S test for E-R graph, average v.s. weighted interaction rule under Beta(2,2) distribution, $p$-value $< 0.01$ ($^*$)}
    \label{tab:K-Sbeta2,2}
\end{table}

\begin{table}[H]
    \centering
    \begin{tabular}{|c|c|c|}
    \hline
    Probability & Strategies & Statistic \\
    \hline
    \multirow{3}{*}{4 $\times 10^{-6}$} & Independent & 0.000\\
    \cline{2-3}
    & Random & 0.002\\
    \cline{2-3}
    & Cluster & 0.010\\
    \hline
    \multirow{3}{*}{8.9 $\times 10^{-5}$} & Independent & 0.000\\
    \cline{2-3}
    & Random & 0.010\\
    \cline{2-3}
    & Cluster & 0.032\\
    \hline
    \multirow{3}{*}{0.002} & Independent & 0.000\\
    \cline{2-3}
    & Random & 0.034\\
    \cline{2-3}
    & Cluster & 0.118$^*$\\
    \hline
    \multirow{3}{*}{0.045} & Independent & 0.000\\
    \cline{2-3}
    & Random & 0.272$^*$\\
    \cline{2-3}
    & Cluster & 0.356$^*$\\
    \hline
    \multirow{3}{*}{0.211} & Independent & 0.000\\
    \cline{2-3}
    & Random & 0.400$^*$\\
    \cline{2-3}
    & Cluster & 0.452$^*$\\
    \hline
    \multirow{3}{*}{0.460} & Independent & 0.000\\
    \cline{2-3}
    & Random & 0.536$^*$\\
    \cline{2-3}
    & Cluster & 0.596$^*$\\
    \hline
    \end{tabular}
    \caption{Kolmogorov-D statistic from 2-sample K-S test for E-R graph, average v.s. weighted interaction rule under Beta(2,5) distribution, $p$-value $< 0.01$ ($^*$)}
    \label{tab:K-Sbeta2,5}
\end{table}

\begin{table}[H]
    \centering
    \begin{tabular}{|c|c|c|}
    \hline
    Probability & Strategies & Statistic \\
    \hline
    \multirow{3}{*}{4 $\times 10^{-6}$} & Independent & 0.000\\
    \cline{2-3}
    & Random & 0.002\\
    \cline{2-3}
    & Cluster & 0.006\\
    \hline
    \multirow{3}{*}{8.9 $\times 10^{-5}$} & Independent & 0.000\\
    \cline{2-3}
    & Random & 0.010\\
    \cline{2-3}
    & Cluster & 0.038\\
    \hline
    \multirow{3}{*}{0.002} & Independent & 0.000\\
    \cline{2-3}
    & Random & 0.044\\
    \cline{2-3}
    & Cluster & 0.136$^*$\\
    \hline
    \multirow{3}{*}{0.045} & Independent & 0.000\\
    \cline{2-3}
    & Random & 0.270$^*$\\
    \cline{2-3}
    & Cluster & 0.304$^*$\\
    \hline
    \multirow{3}{*}{0.211} & Independent & 0.000\\
    \cline{2-3}
    & Random & 0.382$^*$\\
    \cline{2-3}
    & Cluster & 0.448$^*$\\
    \hline
    \multirow{3}{*}{0.460} & Independent & 0.000\\
    \cline{2-3}
    & Random & 0.540$^*$\\
    \cline{2-3}
    & Cluster & 0.542$^*$\\
    \hline
    \end{tabular}
    \caption{Kolmogorov-D statistic from 2-sample K-S test for E-R graph, average v.s. weighted interaction rule under standard normal distribution, $p$-value $< 0.01$ ($^*$)}
    \label{tab:K-Sstandardnormal}
\end{table}

\begin{table}[H]
    \centering
    \resizebox{\columnwidth}{!}{%
    \begin{tabular}{|c|c|c|c|c|}
    \hline
    Distribution & Interaction Rule & Independent & Random & Cluster \\
    \hline
    \multirow{2}{*}{Beta(2,2)} & Average & 0.037 & 0.043 & 0.083 \\
    \cline{2-5}
     & Weighted & 0.037 & 0.041 & 0.072 \\
    \hline
    \multirow{2}{*}{Beta(2,5)} & Average & 0.027 & 0.031 & 0.058 \\
    \cline{2-5}
     & Weighted & 0.027 & 0.030 & 0.051 \\
    \hline
    \multirow{2}{*}{Normal(0,1)} & Average & 0.180 & 0.196 & 0.307 \\
    \cline{2-5}
     & Weighted & 0.180 & 0.195 & 0.311 \\
    \hline
    \end{tabular}}
    \caption{Mean $W_1$ distance for S-F graphs across various scenarios}
    \label{tab:SFgraph}
\end{table}

\begin{table}[H]
    \centering
    \begin{tabular}{|c|c|c|}
    \hline
    \makecell{Initial \\ Distribution} & Strategies & Statistic \\
    \hline
    \multirow{3}{*}{Beta(2,2)} & Independent & 0.000\\
    \cline{2-3}
    & Random & 0.076\\
    \cline{2-3}
    & Cluster & 0.308$^{*}$\\
    \hline
    \multirow{3}{*}{Beta(2,5)} & Independent & 0.000\\
    \cline{2-3}
    & Random & 0.088$^\Delta$\\
    \cline{2-3}
    & Cluster & 0.252$^{*}$\\
    \hline
    \multirow{3}{*}{Normal(0,1)} & Independent & 0.000\\
    \cline{2-3}
    & Random & 0.044\\
    \cline{2-3}
    & Cluster & 0.044\\
    \hline
    \end{tabular}
    \caption{Kolmogorov-D statistic from 2-sample K-S test for S-F graph, average v.s. weighted interaction rule, $p$-value $< 0.05$ ($^\Delta$), $< 0.01$ ($^{*}$)}
    \label{tab:K-SSF}
\end{table}

\begin{figure}[H]
     \centering
     \begin{subfigure}[b]{0.15\textwidth}
         \centering
         \includegraphics[width=\textwidth]{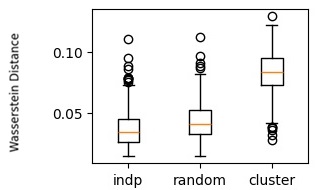}
         \caption{ }
         \label{fig:SF22}
     \end{subfigure}
     \hfill
     \begin{subfigure}[b]{0.15\textwidth}
         \centering
         \includegraphics[width=\textwidth]{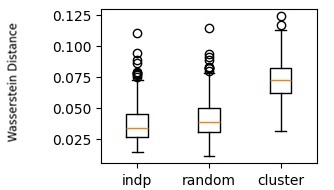}
         \caption{ }
         \label{fig:SF22weighted}
     \end{subfigure}
     \hfill
     \begin{subfigure}[b]{0.15\textwidth}
         \centering
         \includegraphics[width=\textwidth]{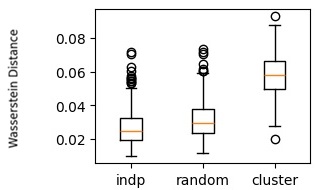}
         \caption{ }
         \label{fig:SF25}
     \end{subfigure}
     \begin{subfigure}[b]{0.15\textwidth}
         \centering
         \includegraphics[width=\textwidth]{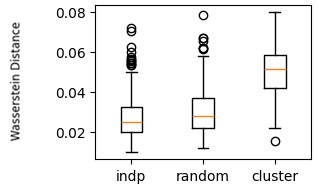}
         \caption{ }
         \label{fig:SF25weighted}
     \end{subfigure}
     \hfill
     \begin{subfigure}[b]{0.15\textwidth}
         \centering
         \includegraphics[width=\textwidth]{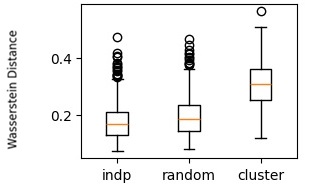}
         \caption{ }\label{fig:SFstandardnormal}
     \end{subfigure}
     \hfill
     \begin{subfigure}[b]{0.15\textwidth}
         \centering
         \includegraphics[width=\textwidth]{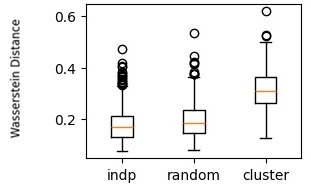}
         \caption{} \label{fig:SFstandardnormalweighted}
     \end{subfigure}
     \caption{S-F model}        \label{fig:overallSF}
\end{figure}

\begin{table}[H]
    \centering
    \resizebox{\columnwidth}{!}{%
    \begin{tabular}{|c|c|c|c|c|c|}
    \hline
    Distribution & \makecell{Interaction \\ Rule} & Independent & Random & Cluster & \makecell{Independent \\ (30)}\\
    \hline
    \multirow{2}{*}{Beta(2,2)} & Average & 0.013 & 0.083 & 0.143  & 0.048 \\
    \cline{2-6}
     & Weighted & 0.013 & 0.073 & 0.136  & 0.048 \\
    \hline
    \multirow{2}{*}{Beta(2,5)} & Average & 0.010 & 0.056 & 0.099  & 0.035 \\
    \cline{2-6}
     & Weighted & 0.010 & 0.050 & 0.094  & 0.035 \\
    \hline
    \multirow{2}{*}{Normal(0,1)} & Average & 0.063 & 0.304 & 0.574 & 0.230 \\
    \cline{2-6}
     & Weighted & 0.063 & 0.302 & 0.571  & 0.230 \\
    \hline
    \end{tabular}}
    \caption{Mean $W_1$ distance for undirected ego-Facebook dataset}
    \label{tab:fbundirected}
\end{table}

\begin{table}[H]
    \centering
    \begin{tabular}{|c|c|c|}
    \hline
    Probability & Strategies & Statistic \\
    \hline
    \multirow{3}{*}{Beta(2,2)} & Independent & 0.000\\
    \cline{2-3}
    & Random & 0.530$^*$\\
    \cline{2-3}
    & Cluster & 0.262$^*$\\
    \hline
    \multirow{3}{*}{Beta(2,5)} & Independent & 0.000\\
    \cline{2-3}
    & Random & 0.508$^*$\\
    \cline{2-3}
    & Cluster & 0.242$^*$\\
    \hline
    \multirow{3}{*}{Normal(0,1)} & Independent & 0.000\\
    \cline{2-3}
    & Random & 0.058\\
    \cline{2-3}
    & Cluster & 0.086$^\Delta$\\
    \hline
    \end{tabular}
    \caption{Kolmogorov-D statistic from 2-sample K-S test for undirected ego-Facebook dataset, average v.s. weighted interaction rule, $p$-value $< 0.05$ ($^\Delta$), $< 0.01$ ($^*$)}
    \label{tab:K-Sfacebook}
\end{table}

\begin{table}[H]
    \centering
    \resizebox{\columnwidth}{!}{%
    \begin{tabular}{|c|c|c|c|c|c|}
    \hline
    Distribution & \makecell{Interaction \\ Rule} & Independent & Random & Cluster & \makecell{Independent \\ (30)}\\
    \hline
    \multirow{2}{*}{Beta(2,2)} & Average & 0.029 & 0.084 & 0.126  & 0.048 \\
    \cline{2-6}
     & Weighted & 0.029 & 0.068 & 0.117 & 0.048 \\
    \hline
    \multirow{2}{*}{Beta(2,5)} & Average & 0.021 & 0.058 & 0.087 & 0.035 \\
    \cline{2-6}
     & Weighted & 0.021 & 0.047 & 0.081 & 0.035 \\
    \hline
    \multirow{2}{*}{Normal(0,1)} & Average & 0.144 & 0.287 & 0.496  & 0.230 \\
    \cline{2-6}
     & Weighted & 0.144 & 0.286 & 0.495 & 0.230 \\
    \hline
    \end{tabular}}
    \caption{Mean $W_1$ distance for directed EU-email dataset}
    \label{tab:emaildirected}
\end{table}

\begin{table}[H]
    \centering
    \begin{tabular}{|c|c|c|}
    \hline
    \makecell{Initial \\ Distribution} & Strategies & Statistic \\
    \hline
    \multirow{4}{*}{Beta(2,2)} & Independent & 0.000\\
    \cline{2-3}
    & Random & 0.472$^*$\\
    \cline{2-3}
    & Cluster & 0.214$^*$\\
    \cline{2-3}
    & Independent (10) & 0.000\\
    \hline
    \multirow{4}{*}{Beta(2,5)} & Independent & 0.000\\
    \cline{2-3}
    & Random & 0.402$^*$\\
    \cline{2-3}
    & Cluster & 0.180$^*$\\
    \cline{2-3}
    & Independent (10) & 0.000\\
    \hline
    \multirow{4}{*}{Normal(0,1)} & Independent & 0.000\\
    \cline{2-3}
    & Random & 0.034\\
    \cline{2-3}
    & Cluster & 0.062\\
    \cline{2-3}
    & Independent (10) & 0.000\\
    \hline
    \end{tabular}
    \caption{Kolmogorov-D statistic from 2-sample K-S test for directed EU-email dataset, $p$-value $< 0.01$ ($^*$)}
    \label{tab:K-Sdirectedemail}
\end{table}

\begin{figure}[H]
     \centering
     \begin{subfigure}[b]{0.2\textwidth}
         \centering
         \includegraphics[width=\textwidth]{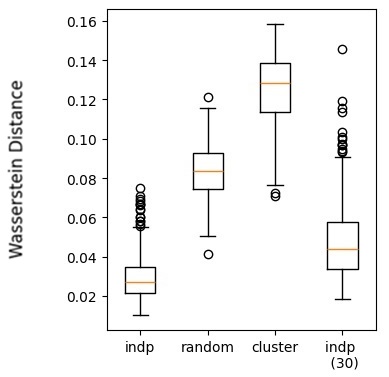}
         \caption{ }
         \label{fig:email22}
     \end{subfigure}
     \begin{subfigure}[b]{0.2\textwidth}
         \centering
         \includegraphics[width=\textwidth]{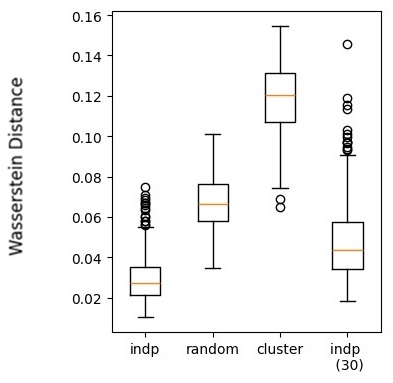}
         \caption{ }
         \label{fig:email22weighted}
     \end{subfigure}
     \begin{subfigure}[b]{0.2\textwidth}
         \centering
         \includegraphics[width=\textwidth]{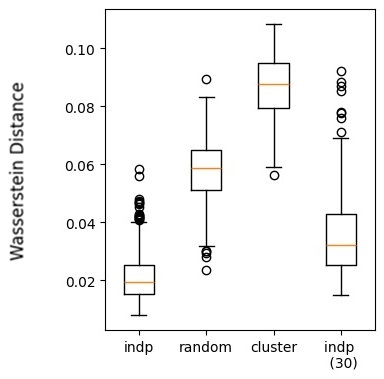}
         \caption{ }
         \label{fig:email25}
     \end{subfigure}
     \begin{subfigure}[b]{0.2\textwidth}
         \centering
         \includegraphics[width=\textwidth]{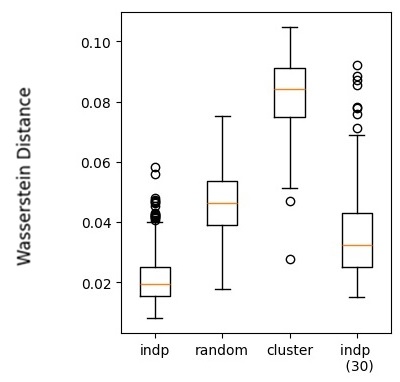}
         \caption{ }
         \label{fig:email25weighted}
     \end{subfigure}
     \begin{subfigure}[b]{0.207\textwidth}
         \centering
         \includegraphics[width=\textwidth]{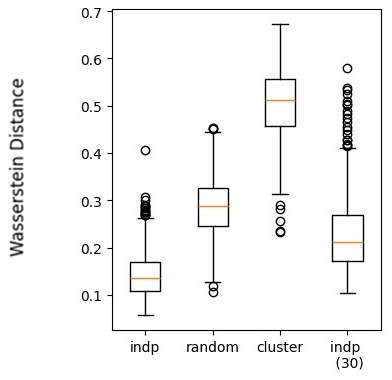}
         \caption{ }\label{fig:emailstandardnormal}
     \end{subfigure}
     \begin{subfigure}[b]{0.2\textwidth}
         \centering
         \includegraphics[width=\textwidth]{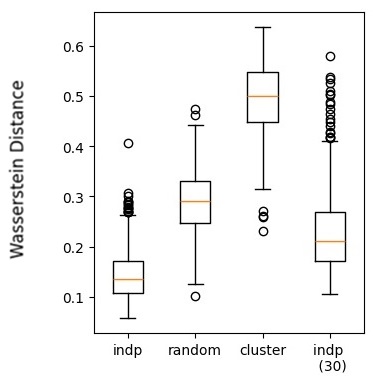}
         \caption{} \label{fig:emailstandardnormalweighted}
     \end{subfigure}
     \caption{Using \code{ego-Eu-core} dataset from SNAP}        \label{fig:overallemail}
\end{figure}

\begin{table}[H]
    \centering
    \resizebox{\columnwidth}{!}{%
    \begin{tabular}{|c|c|c|c|c|c|}
    \hline
    Distribution & \makecell{Interaction \\ Rule} & Independent & Random & Cluster  & \makecell{Independent \\ (30)}\\
    \hline
    \multirow{2}{*}{Beta(2,2)} & Average & 0.027 & 0.069 & 0.122 & 0.048 \\
    \cline{2-6}
     & Weighted & 0.027 & 0.065 & 0.116  & 0.048 \\
    \hline
    \multirow{2}{*}{Beta(2,5)} & Average & 0.019 & 0.047 & 0.084 & 0.035 \\
    \cline{2-6}
     & Weighted & 0.019 & 0.044 & 0.080  & 0.035 \\
    \hline
    \multirow{2}{*}{Normal(0,1)} & Average & 0.132 & 0.276 & 0.485 & 0.230 \\
    \cline{2-6}
     & Weighted & 0.132 & 0.275 & 0.481  & 0.230 \\
    \hline
    \end{tabular}}
    \caption{Mean $W_1$ distance for undirected EU-email dataset}
    \label{tab:emailundirected}
\end{table}

\begin{table}[H]
    \centering
    \begin{tabular}{|c|c|c|}
    \hline
    \makecell{Initial \\ Distribution} & Strategies & Statistic \\
    \hline
    \multirow{4}{*}{Beta(2,2)} & Independent & 0.000\\
    \cline{2-3}
    & Random & 0.144$^*$\\
    \cline{2-3}
    & Cluster & 0.152$^*$\\
    \cline{2-3}
    & Independent (10) & 0.000\\
    \hline
    \multirow{4}{*}{Beta(2,5)} & Independent & 0.000\\
    \cline{2-3}
    & Random & 0.134$^*$\\
    \cline{2-3}
    & Cluster & 0.172$^*$\\
    \cline{2-3}
    & Independent (10) & 0.000\\
    \hline
    \multirow{4}{*}{Normal(0,1)} & Independent & 0.000\\
    \cline{2-3}
    & Random & 0.040\\
    \cline{2-3}
    & Cluster & 0.070\\
    \cline{2-3}
    & Independent (10) & 0.000\\
    \hline
    \end{tabular}
    \caption{Kolmogorov-D statistic from 2-sample K-S test for undirected EU-email dataset, average v.s. weighted interaction rule, $p$-value $< 0.01$ ($^*$)}
    \label{tab:K-Sundirectedemail}
\end{table}

\begin{table}[H]
    \centering
    \resizebox{\columnwidth}{!}{%
    \begin{tabular}{|c|c|c|c|c|c|}
    \hline
    Distribution & \makecell{Interaction \\ Rule} & Independent & Random & Cluster & \makecell{Independent \\ (30)}\\
    \hline
    \multirow{2}{*}{Beta(2,2)} & Average & 0.122$^*$ & 0.466$^*$ & 0.180$^*$  & 0 \\
    \cline{2-6}
     & Weighted & 0.122$^*$ & 0.128$^*$ & 0.096$^\Delta$  & 0 \\
    \hline
    \multirow{2}{*}{Beta(2,5)} & Average & 0.148$^*$ & 0.398$^*$ & 0.170$^*$  & 0 \\
    \cline{2-6}
     & Weighted & 0.148$^*$ & 0.116$^*$ & 0.110$^*$  & 0 \\
    \hline
    \multirow{2}{*}{Normal(0,1)} & Average & 0.134$^*$ & 0.126$^*$ & 0.108$^*$  & 0 \\
    \cline{2-6}
     & Weighted & 0.134$^*$ & 0.100$^\Delta$ & 0.142$^*$  & 0 \\
    \hline
    \end{tabular}}
    \caption{Kolmogorov-D statistic from 2-sample K-S test for directed against undirected EU-email dataset, $p$-value $< 0.05$ ($^\Delta$), $< 0.01$ ($^*$)}
    \label{tab:kstestemail}
\end{table}

%








\end{document}